\documentclass{amsart}
\usepackage{amsfonts}
\usepackage{amsmath,amscd,lscape}
\usepackage{amsthm}
\usepackage{amssymb}
\usepackage{latexsym}
\usepackage{mathrsfs}
\usepackage{ulem}
\usepackage[utf8]{inputenc}
\usepackage{cancel}
\usepackage{tikz}
\usepackage[linktocpage=true]{hyperref}

\usepackage[utf8]{inputenc}

\usepackage{graphicx}
\usepackage{subcaption}
\usepackage{float}

\newcommand{\nc}{\newcommand}
\nc{\ben}{\begin{eqnarray}}
\nc{\een}{\end{eqnarray}}

\newcommand{\beqa}{\begin{eqnarray}}
\newcommand{\eeqa}{\end{eqnarray}}
\nc{\Z}{{\bold Z}}

\newcommand{\fpt}[7]{{}_4\phi_3\left[\begin{matrix} #1 , #2, #3, #4 \\
#5, #6, #7 \end{matrix}\,; q^2,q^2\right]}

\usepackage[curve,matrix,arrow,color]{xy}

\usepackage{color}

\newtheorem{cor}{Corollary}[section]
\newtheorem{lem}{Lemma}[section]

\newtheorem{prop}{Proposition}[section]

\newtheorem{example}{Example}
\newtheorem{defn}{Definition}[section]
\newtheorem{thm}{Theorem}

\newtheorem{rem}{Remark}

\newtheorem{hyp}{Hypothesis}

\newcommand{\cal}{\mathcal}

\newcommand{\tA}{\textsf{A}}

\newcommand{\tW}{\textsf{W}}

\newcommand{\tb}{\mathsf{b}} 
\newcommand{\tc}{\mathsf{c}}

\newcommand{\cV}{\mathcal{V}}

\newcommand{\non}{\nonumber}

\numberwithin{equation}{section}
\begin{document}

\title[The $q$-Racah polynomials from scalar products of Bethe states II
]{The $q$-Racah polynomials from scalar products of Bethe states II
}
\author{Pascal Baseilhac$^{*}$}
\address{$^*$ Institut Denis-Poisson CNRS/UMR 7013 - Universit\'e de Tours - Universit\'e d'Orl\'eans
Parc de Grammont, 37200 Tours, 
FRANCE}
\email{pascal.baseilhac@univ-tours.fr}

\author{Rodrigo A. Pimenta$^{**}$}
\address{$^{**}$ Instituto de F\'isica - Universidade Federal do Rio Grande do Sul, Porto Alegre, 91501-970, BRAZIL} 
\email{rodrigo.pimenta@ufrgs.br}

\begin{abstract} The theory of Leonard triples is applied to the derivation of normalized scalar products of on-shell and off-shell Bethe states generated from  a Leonard pair. The scalar products take the form of linear combinations  of $q$-Racah polynomials with coefficients depending on the off-shell parameters. Upon specializations, explicit solutions for the corresponding Belliard-Slavnov linear systems are obtained. It implies  the existence of a determinant formula in terms of  inhomogeneous Bethe roots for the $q$-Racah polynomials. Also, a set of relations that   determines   solutions (Bethe roots) of the corresponding Bethe equations of inhomogeneous type in terms of solutions of Bethe equations of homogenous type is obtained.
\end{abstract}

\maketitle

	\vskip -0.5cm
	


\section{Introduction}
In a recent series of papers \cite{BP19,Nico1,Nico2,BP22}, a correspondence between the representation theory of the Askey-Wilson  algebra \cite{T87,Z91} or  Racah  algebra \cite{LL,GZ88} and the framework of the algebraic Bethe ansatz \cite{STF79,Skly88}  has been investigated. More precisely,  this correspondence aims to relate  two subjects that have been independently  developed in the context of representation theory and  mathematical physics, respectively. 
On one hand, irreducible finite dimensional representations of the Askey-Wilson or Racah algebras have been extensively studied in recent years. They are  classified using the theory of Leonard pairs \cite{T04,Ter04}. In this context, it is shown that 
the entries of the transition matrices relating different eigenbases of elements of a Leonard pair correspond to discrete orthogonal polynomials that belong to the Askey-scheme \cite{T03}.  
On the other hand, the framework of the algebraic Bethe ansatz has a long history (for a review and references therein, see e.g. \cite{S22}) and provides an efficient approach for deriving exact results in a large class of quantum  integrable systems. In order to solve quantum integrable models with  $U(1)$ symmetry breaking boundary fields, the original approach is upgraded to the {\it modified} algebraic Bethe ansatz approach developed in  \cite{BC13,B15,C15,ABGP15,BP15}.   In both approaches, the spectral problem  for the Hamiltonian of the system  (eigenvalues and eigenvectors)  is formulated in terms of Bethe roots satisfying a system of transcendental equations known as the Bethe equations whereas the eigenvectors are Bethe states.  Depending on the system under consideration,  the Bethe equations are  either of  homogeneous type or of  inhomogeneous type \cite{ODBAbook}. In particular, some simple examples of integrable systems\footnote{The $q$-analog of the quantum Euler top and the `1-site' open XXZ spin chain \cite[Subsection 5.1]{BP19}; A 3-sites Heisenberg chain with inhomogeneous couplings \cite[Subsection 5.2]{BP19}.} and corresponding Bethe equations can be studied using the theory of Leonard pair \cite{BP19}.
Thus, establishing a Leonard pair/algebraic Bethe ansatz correspondence aims to get new insights in each subject.  Also, having in mind that the theory of Leonard pairs generalizes to the theory of tridiagonal pairs \cite{Ter03}, a Leonard pair/algebraic Bethe ansatz correspondence can be viewed as a toy model for a tridiagonal pair/algebraic Bethe ansatz correspondence. The latter should find applications in the theory of special functions generalizing the Askey-scheme and quantum integrable systems with non-trivial boundary conditions and finite dimensional quantum space. More details are given in the last section. \vspace{1mm}

The starting point of the Leonard pair/algebraic Bethe ansatz correspondence is the observation  \cite{BP19} that the  universal exchange relations which typically arise in the algebraic Bethe ansatz analysis of various examples of quantum integrable spin chains with boundaries   are satisfied by so-called dynamical operators generated from the Askey-Wilson algebra. A similar observation holds for  the Racah algebra \cite{Nico2}. This opens the possibility to reconsider the main objects of the theory of Leonard pair (eigenbases, transition coefficients, discrete orthogonal polynomials of the Askey-scheme,...) from the perspective of the   algebraic Bethe ansatz and its modified version  (on-shell Bethe states and related  Bethe equations, scalar products of Bethe states, Baxter Q-polynomial,...) and reciprocally.\vspace{1mm}

For instance, in \cite{BP19} the spectral problem for certain combinations of elements of a Leonard pair of q-Racah type - known as the Heun-Askey-Wilson operator - has been solved using the theory of Leonard pairs combined with the framework of the modified algebraic Bethe ansatz \cite{BP19}. For the limiting case $q=1$ associated with the Racah algebra and related Leonard pairs, the analog analysis is done in \cite{Nico1,Nico2}. In \cite{BP22}, different types of eigenbases and dual eigenbases of Leonard pairs of q-Racah type are constructed using the framework of the modified algebraic Bethe ansatz:  it is shown that eigenbases for a Leonard pair can be built from Bethe states associated with either homogeneous or inhomogeneous Bethe equations. This lead to an interpretation of the $q$-Racah polynomials\footnote{All other families of orthogonal polynomials of the scheme can be reached from the $q$-Racah polynomials by various limit transitions (either in the scalar parameters entering into the definition of the polynomials, or $q\rightarrow 1$). See \cite{KS96,Ko10} for details.} as certain ratios of scalar products of various types of on-shell Bethe states.\vspace{1mm}

The present paper is the second part of \cite{BP22}. Here, we  investigate further the Leonard pair/algebraic Bethe ansatz correspondence motivated by the following problems.   \vspace{1mm}

(a) In \cite{BP22}, it was shown that
certain ratios of  scalar products of {\it on-shell}  versus {\it on-shell} Bethe states are identical to $q$-Racah polynomials. A first goal of this paper is to extend the analysis to the computation of scalar products of on-shell versus {\it off-shell} Bethe states. The strategy is to use the theory of Leonard pairs \cite{Ter01,T03,Ter04,T04} and Leonard triples \cite{Curt07} (see also \cite{HuangLT}).\vspace{1mm}

(b)  For some class of scalar products of on-shell versus off-shell Bethe states, compact determinant formulas involving Bethe roots can be derived by solving certain systems of linear equations  constructed by Belliard and Slavnov \cite{BS19}. A second goal of this paper is to identify the Belliard-Slavnov system of linear equations satisfied by some specialization of the scalar products computed in (a). This aims to derive a new presentation of $q$-Racah polynomials using determinants of matrices associated with inhomogenous Bethe roots.\vspace{1mm}

(c) In the modified algebraic Bethe ansatz framework, inhomogeneous Bethe equations arise.  Contrary to the standard algebraic Bethe ansatz where Bethe roots are classified using an integer $N$, there is no known criteria to classify corresponding Bethe roots. The third goal of this paper is to identify some criteria for the class of inhomogenous Bethe equations under consideration.\vspace{1mm}

Note that (a), (b) and (c) can be viewed as a toy model for studying scalar products of Bethe states in the open XXZ spin chain from the perspective of the theory of tridiagonal pairs, see more details in the 
Concluding Remarks section. \vspace{1mm}

The text is organized as follows. In Section 2, some known results on the Leonard pair of $q$-Racah type/algebraic Bethe ansatz correspondence are reviewed based on \cite{BP22}. The exchange relations associated with a Leonard pair $(A,A^*)$ and corresponding dynamical operators that generate off-shell Bethe states are briefly recalled. Then, eigenbases for $(A,A^*)$ built from on-shell Bethe states  associated with either homogeneous or inhomogeneous Bethe equations are given. 
In Section 3, the goal is to compute explicitly {\it normalized} scalar products of on-shell Bethe states versus off-shell Bethe states using the Leonard pairs data, in order to extend the results of \cite{BP22}. To this end, the Leonard triple of $q$-Racah type  $(A,A^*,A^\diamond)$ is introduced. Eigenbases for the triple $(A,A^*,A^\diamond)$ and related transition matrices are given explicitly. Using the triple, the normalized scalar products are  computed, see Theorem \ref{prop:scalprod}. Upon specializations of previous results, a relation characterizing the  Bethe roots of inhomogenous type associated with the eigenvalue $\theta^*_N$, $N=0,....,2s$ in terms of Bethe roots of homogeneous type associated with the same eigenvalue $\theta^*_N$ is obtained, see Proposition \ref{relinhhom}. In Section 4, the Belliard-Slavnov  system of linear equations satisfied by certain specializations of the normalized scalar products of Theorem \ref{prop:scalprod} are derived, see Proposition \ref{prop:lineq}. As shown in Proposition \ref{prop:scaldet}, a determinant formula exist for such normalized scalar products. Upon specialization to on-shell Bethe states, it implies the existence of a determinant formula for the $q$-Racah polynomials, see Corollary \ref{qRacdet}. A summary of the results and some perspectives are given in the last section. Some formulas and explicit examples are reported in Appendix \ref{apA} and \ref{apB}, respectively.\vspace{2mm}

{\bf Notations:} The parameter $q$ is assumed not to be a root of unity. We write $[X,Y]_q = qXY - q^{-1}YX$ and $[n]_q=\frac{q^n-q^{-n}}{q-q^{-1}}$. 
We use the standard $q$-shifted factorials:
\beqa
(a;q)_n = \prod_{k=0}^{n-1}(1-aq^k)\ ,\qquad (a_1,a_2,\cdots,a_k;q)_n = \prod_{j=1}^{k}(a_j;q)\ ,
\eeqa
and define
\beqa
b(x)=x-x^{-1}\,\,.\label{b} 
\eeqa

\section{Algebraic Bethe ansatz and Leonard pairs of $q$-Racah type}
In this section, the correspondence between the theory of Leonard pairs of $q$-Racah type and the algebraic Bethe ansatz formalism is   reviewed (see  \cite{BP19,BP22} and references therein). Off-shell Bethe states are built using the image of dynamical operators satisfying the exchange relations associated with the Askey-Wilson algebra, and  eigenbases for each operator of the Leonard pair $(A,A^*)$ are given in the form of on-shell Bethe states associated either with Bethe equations of homogeneous type or Bethe equations of inhomogeneous type.    

\subsection{Dynamical operators and the Askey-Wilson algebra}
In the context of the quantum inverse scattering method,  the entries $ \{\mathscr{A}(u),\mathscr{B}(u),\mathscr{C}(u),\mathscr{D}(u)\}$ with $u\in{\mathbb C}^*$ of the so-called (fundamental spin-$\frac{1}{2}$) monodromy matrix or double-row monodromy matrix are the key ingredients in the analysis of transfer matrix' spectral problem, scalar products of Bethe states, correlation functions and form factors. For a review, see \cite{S22}. Importantly, according to the integrable model under consideration those entries satisfy a set of exchange relations that are equivalent to either a Yang-Baxter equation or a reflection equation. 

In order to implement a solution based on the algebraic Bethe ansatz approach to some of the fundamental examples of quantum integrable models (see e.g. \cite{TF79,CLSW03}) the analysis however requires to introduce certain linear combinations of those entries, which are called {\it dynamical operators}. Let us introduce the fundamental dynamical operators
\beqa
\{\mathscr{A}^{\epsilon}(u,m), \mathscr{B}^{\epsilon}(u,m),\mathscr{C}^{\epsilon}(u,m),\mathscr{D}^{\epsilon}(u,m)|m \in {\mathbb N}^*,\epsilon= \pm 1\}
\eeqa
subject to the set of universal  exchange relations\footnote{Those relations have been first introduced in the analysis of the open XXZ spin-$1/2$ chain with generic integrable boundary conditions \cite{CLSW03} In this case,  the dynamical operators specialize to operators acting on finite dimensional tensor product representations of $U_q(\widehat{sl_2})$.}:
\ben
\qquad \mathscr{B}^{\epsilon}(u,m+2)\mathscr{B}^{\epsilon}(v,m) &=& \mathscr{B}^{\epsilon}(v,m+2)\mathscr{B}^{\epsilon}(u,m),\label{comBdBd} \\
\qquad \mathscr{A}^{\epsilon}(u,m+2)\mathscr{B}^{\epsilon}(v,m)&=&f(u,v)\mathscr{B}^{\epsilon}(v,m)\mathscr{A}^{\epsilon}(u,m) \label{comAdBd}\\
&& + g(u,v,m)\mathscr{B}^{\epsilon}(u,m)\mathscr{A}^{\epsilon}(v,m) + w(u,v,m)\mathscr{B}^{\epsilon}(u,m)\mathscr{D}^{\epsilon}(v,m),\nonumber  \\
\mathscr{D}^{\epsilon}(u,m+2)\mathscr{B}^{\epsilon}(v,m)&=& h(u,v)\mathscr{B}^{\epsilon}(v,m) \mathscr{D}^{\epsilon}(u,m) 
,\label{comDdBd}\\
&&+k(u,v,m)\mathscr{B}^{\epsilon}(u,m)\mathscr{D}^{\epsilon}(v,m)+ n(u,v,m)\mathscr{B}^{\epsilon}(u,m)\mathscr{A}^{\epsilon}(v,m),\nonumber\\
\mathscr{C}^{\epsilon}(u,m+2)\mathscr{B}^{\epsilon}(v,m)&=& \mathscr{B}^{\epsilon}(v,m-2)\mathscr{C}^{\epsilon}(u,m)  \label{comcdBd} \\
&& +q(u,v,m) \mathscr{A}^{\epsilon}(v,m)\mathscr{D}^{\epsilon}(u,m)+r(u,v,m)\mathscr{A}^{\epsilon}(u,m)\mathscr{D}^{\epsilon}(v,m)
\nonumber\\
&& +s(u,v,m) \mathscr{A}^{\epsilon}(u,m)\mathscr{A}^{\epsilon}(v,m)+x(u,v,m)\mathscr{A}^{\epsilon}(v,m)\mathscr{A}^{\epsilon}(u,m)
\nonumber\\
&&+y(u,v,m) \mathscr{D}^{\epsilon}(u,m)\mathscr{A}^{\epsilon}(v,m)+z(u,v,m)\mathscr{D}^{\epsilon}(u,m)\mathscr{D}^{\epsilon}(v,m) \nonumber
\een
and
\ben
\qquad \mathscr{C}^{\epsilon}(u,m-2)\mathscr{C}^{\epsilon}(v,m) &=& \mathscr{C}^{\epsilon}(v,m-2)\mathscr{C}^{\epsilon}(u,m),\label{comCdCd} \\
\qquad  \mathscr{C}^{\epsilon}(v,m+2)\mathscr{A}^{\epsilon}(u,m+2)&=&f(u,v)\mathscr{A}^{\epsilon}(u,m)\mathscr{C}^{\epsilon}(v,m+2)  \label{comAdCd}
\\ && +  g(u,v,m)\mathscr{A}^{\epsilon}(v,m)\mathscr{C}^{\epsilon}(u,m+2) + w(v,u,m)\mathscr{D}^{\epsilon}(v,m)\mathscr{C}^{\epsilon}(u,m+2), \nonumber \\
\qquad  \mathscr{C}^{\epsilon}(v,m+2)\mathscr{D}^{\epsilon}(u,m+2)&=& h(u,v)\mathscr{D}^{\epsilon}(u,m) \mathscr{C}^{\epsilon}(v,m+2) \label{comDdCd}
\\ &&+k(u,v,m)\mathscr{D}^{\epsilon}(v,m)\mathscr{C}^{\epsilon}(u,m+2)+ n(u,v,m)\mathscr{A}^{\epsilon}(v,m)\mathscr{C}^{\epsilon}(u,m+2)\,,\nonumber
\een
where the explicit expressions for the coefficients $f(u,v),g(u,v,m),...$ are reported in Appendix \ref{apA}. \vspace{1mm}

The above set of exchange relations is {\it universal} in the following sense. It follows from a reflection equation - up to conventions, see \cite[eq. (12)]{CLSW03} - which provides a presentation for the alternating central extension of the $q$-Onsager algebra ${\cal A}_q$ \cite{BSh1,Ter21}. By construction, the corresponding dynamical operators are power series of the inverse spectral parameter $u^{-1}$ (or $u$ depending on the convention chosen) with coefficients in ${\cal A}_q$. For several examples of quantum integrable models such as the open XXZ spin-$j$  chain \cite{D07}, the trigonometric solid-on-solid model \cite{FK10} , the alternating spin chain \cite{CYSW14},
the relativistic Toda chain \cite{ZCYSW17}, the lattice sine-Gordon model \cite{MNP17}  or the three-sites interacting model constructed in \cite[Section 5.2]{BP19}, the fundamental exchange relations take the form (\ref{comBdBd})-(\ref{comDdCd}) where the dynamical operators $\{\mathscr{A}^{\epsilon}(u,m), \mathscr{B}^{\epsilon}(u,m),\mathscr{C}^{\epsilon}(u,m),\mathscr{D}^{\epsilon}(u,m)$ specialize to certain Laurent polynomials in $u$ with coefficients in a certain quotient of ${\cal A}_q$.\vspace{1mm}

In continuity with the previous works \cite{BP19,BP22}, in this paper we consider the quotient of ${\cal A}_q$ known as the Askey-Wilson algebra.
 The  Askey-Wilson algebra (AW) is generated by $\tA,\tA^*$ subject to the relations \cite{T87,Z91}
 \beqa
 &&\big[\tA,\big[\tA,\tA^*\big]_q\big]_{q^{-1}}=  \rho \,\tA^*+\omega \,\tA+\eta\mathcal{I} \ , \label{aw1} \\
 &&\big[\tA^*,\big[\tA^*,\tA\big]_q\big]_{q^{-1}}= \rho \,\tA+\omega \,\tA^*+\eta^*\mathcal{I} \ ,\label{aw2}
 \eeqa
where the structure constants are $\{\rho,\omega,\eta,\eta^*\}\in{\mathbb C}^*$ and the identity element is denoted $\mathcal{I}$. In this case, the explicit expressions of the dynamical operators that solve the exchange relations (\ref{comBdBd})-(\ref{comDdCd}) are given in \cite[Appendix 4]{BP19}. As Laurent polynomials in $u$ with coefficients in AW, the dynamical operators depend on the free scalars $\alpha,\beta$ (see Appendix \ref{apA}) and $\chi\in{\mathbb C}^*$ (called `gauge parameters'). For instance, one has:
	\ben
&&\quad \mathscr{B}^{+}(u,m)=\frac{\beta  b(u^2)}{\alpha  q^{2 m+2}-\beta }\Big(\frac{\chi  q^m }{\beta  \rho  u}[\textsf{A}^*,\textsf{A}]_q-\frac{\beta  q^{-m} }{u \chi
}[\textsf{A},\textsf{A}^*]_q+\frac{ (q^2+1)}{q
	u}\textsf{A}-(\frac{1}{q u^3}+q u)\textsf{A}^* 
\label{Bm}\\&&\quad
+\frac{q^{-m-3}}{\beta  \rho 
	\chi  b(q^2)}\big((q^2+1) u^{-1} (\omega  (\chi ^2 q^{2 m+2}-\beta ^2 q^2 \rho )+\beta   (q^4-1) \eta^*\chi  q^m)-q^2 \rho  (q^2 u+u^{-3})  (\beta ^2 \rho -\chi ^2
q^{2 m})\big)\Big),\nonumber\\ \nonumber
\een
\ben
&&\mathscr{B}^{-}(u,m)=
\frac{\beta  b(u^2) q^{2 m+1}}{\alpha q^{-2}-\beta  q^{2 m}}
\Big(\frac{u \chi  q^{-m-1} }{\beta 
	\rho }[\textsf{A}^*,\textsf{A}]_q-\frac{\beta  u q^{m-1} }{\chi }[\textsf{A},\textsf{A}^*]_q+\frac{ (q^2 u^4+1)}{q^2
	u}\textsf{A}-\frac{ (q^2+1) u}{q^2}\textsf{A}^*
\nonumber\\&&\quad
-\frac{q^{-m-4} }{\beta  \rho  \chi 
	b(q^2)}\big(\rho  q^2 (q^2 u^3+u^{-1})  (\beta ^2 \rho  q^{2 m}-\chi ^2)+(q^2+1) u
(\omega  (\beta ^2 \rho  q^{2 m+2}-q^2 \chi ^2)+\beta    (q^4-1)\eta \chi 
q^m)\big)\Big)\ . \nonumber
\een

\vspace{1mm}

\subsection{Leonard pairs and off-shell/on-shell Bethe states}\label{sec22}
For most of the known examples of quantum integrable models associated with the exchange relations (\ref{comBdBd})-(\ref{comDdCd}), the dynamical operators act on a finite dimensional vector space. Following \cite{BP19,BP22}, here we consider irreducible finite dimensional representations of AW which relate to the concept of Leonard pair of $q$-Racah type. Eigenbases for a Leonard pair $(A,A^*)$  and their duals are first recalled. Then, following \cite{BP22} off-shell Bethe states associated with the Leonard pair and the correspondence between on-shell Bethe states and eigenvectors of the Leonard pair is given. \vspace{1mm}

\subsubsection{The Leonard pair $(A,A^*)$ and related eigenbases}\label{sec:LP} Let $\cV$ be a vector space of positive finite dimension $\dim({\cV})=2s+1$, where $s$ is an integer or half-integer. Define $\pi: {\rm AW} \rightarrow {\rm End}(\cV)$. Assume 
$\pi(\tA),\pi(\tA^*)$ are diagonalizable on  ${\cV}$, each  multiplicity-free and ${\cV}$  is irreducible. Then,
$\pi(\tA),\pi(\tA^*)$ is a Leonard pair \cite[Definition 1.1]{T04}. For further convenience, we denote:
\beqa
\pi(\tA) = A \ , \qquad \pi(\tA^*) = A^* \ .\label{imA}
\eeqa
 Let $\{\theta_M\}_{M=0}^{2s}$ denote the eigenvalue sequence  associated with   
 $A$ (resp.  $\{\theta_N^*\}_{N=0}^{2s}$ the eigenvalue sequence  associated with   
 $A^*$). By definition of a Leonard pair \cite{Ter01}, there exists an eigenbasis with vectors  $\{|\theta_M \rangle \}_{M=0}^{2s}$ (resp. an eigenbasis with vectors $\{|\theta^*_N\rangle \}_{N=0}^{2s}$)
such that
\beqa
\qquad \quad A |\theta_M\rangle &=& \theta_M|\theta_M\rangle \ ,\quad  A^* |\theta_M\rangle =
A^{(*,.)}_{M+1,M} |\theta_{M+1}\rangle +  A^{(*,.)}_{M,M} |\theta_{M}\rangle +  A^{(*,.)}_{M-1,M}|\theta_{M-1}\rangle \ ,\label{tridAstar}\\
\qquad \quad  A^* |\theta^*_N \rangle &=&  \theta^*_N |\theta^*_N \rangle  \ ,\quad  A|\theta^*_N \rangle   =
A^{(.,*)}_{N+1,N} |\theta^*_{N+1} \rangle  +   A^{(.,*)}_{N,N} |\theta^*_N \rangle  +   A^{(.,*)}_{N-1,N}|\theta^*_{N-1} \rangle \ . \label{tridAstar2}
\eeqa
Here $A^{(*,.)}_{-1,0}=A^{(*,.)}_{2s+1,2s}=A^{(.,*)}_{-1,0}=A^{(.,*)}_{2s+1,2s}=0$ and all other coefficients are non-zero, the representation being irreducible. 
Given a Leonard pair, the transition matrix $P^{\{.,*\}}$ from the basis $\{|\theta_M \rangle \}_{M=0}^{2s}$ to the basis $\{|\theta^*_N \rangle \}_{N=0}^{2s}$ (resp. the inverse transition matrix ${P^{\{.,*\}}}^{-1}$ from the basis $\{|\theta^*_N \rangle \}_{N=0}^{2s}$ to the basis $\{|\theta_M \rangle \}_{M=0}^{2s}$) is known explicitly in terms of $q$-Racah polynomials \cite{Ter04}:
\beqa
|\theta^*_N\rangle = \sum_{M=0}^{2s} P^{\{.,*\}}_{MN}  |\theta_M \rangle \qquad \mbox{and} \qquad |\theta_M\rangle = \sum_{N=0}^{2s}
({P^{\{.,*\}}}^{-1})_{NM}|\theta^*_N \rangle \label{transeq}
\eeqa 
where the coefficients of the transition matrices  are given by:
\beqa
P^{\{.,*\}}_{MN}=k_N R_M(\theta^*_N) \qquad \mbox{and}  \qquad ({P^{\{.,*\}}}^{-1})_{NM} =   \nu_0^{-1} k^*_M R_M(\theta^*_N)\ .\label{PMN}
\eeqa
More details on the explicit expressions for the coefficients $\{A^{(*,.)}_{M\pm 1,M},A^{(*,.)}_{M,M}\}$ and $\{A^{(.,*)}_{N\pm 1,N},A^{(.,*)}_{N,N}\}$, the coefficients $\nu_0,k_N,k_M^*$ and the $q$-Racah polynomials $R_M(\theta_N^*)$ are given in Subsection \ref{sec2}. \vspace{1mm}

Let $\tilde\cV$ be the vector space dual of $\cV$, i.e. the vector space  $\tilde\cV$ of all linear functionals from $\cV$ to ${\mathbb C}$. Define the family of covectors $\{\langle \theta_M |\}_{M=0}^{2s}\in \tilde\cV$ (resp. $\{\langle \theta^*_M |\}_{N=0}^{2s}\in \tilde\cV$) associated with the eigenvalue sequence $\{\theta_M\}_{M=0}^{2s}$ (resp.   $\{\theta_N^*\}_{N=0}^{2s}$). One has:
\beqa
\qquad \quad \langle \theta_M |A  &=& \langle \theta_M |\theta_M\ ,\quad  \langle \theta_M |A^* =\langle \theta_{M+1 }|
\tilde{A}^{(*,.)}_{M+1,M}  +  \langle \theta_{M }|\tilde{A}^{(*,.)}_{M,M}  + \langle \theta_{M-1 } |\tilde{A}^{(*,.)}_{M-1,M} \ ,\label{tridAstardual}\\
\qquad \quad \langle \theta_N^* |A^* &=& \langle \theta_N^* |\theta_N^*\ ,\quad  \langle \theta_N^* |A =\langle \theta_{N+1 }^*|
\tilde{A}^{(.,*)}_{N+1,N}  +  \langle \theta_{N }^*|\tilde{A}^{(.,*)}_{N,N}  + \langle \theta_{N-1 } |\tilde{A}^{(.,*)}_{N-1,N}\  .
\label{tridAstar2dual}
\eeqa
Introduce the scalar product $\langle .|. \rangle : \tilde{\cV} \times \cV \rightarrow {\mathbb C}$.  From the orthogonality relations 
\beqa
\langle \theta_{M'} |\theta_M \rangle  = \delta_{MM'}\xi_M \ ,\qquad \langle \theta_{N'}^* |\theta_N^* \rangle  = \delta_{NN'}\xi_N^* \label{orthcond}\ ,
\eeqa
where $\delta_{ij}$ denotes the Kronecker symbol, one gets:
\beqa
\tilde{A}^{(*,.)}_{M',M} &=&  A^{(*,.)}_{M,M'} \frac{\xi_{M}}{\xi_{M'}}\quad \mbox{for $M'=M+1,M,M-1$}\ , \\
\tilde{A}^{(.,*)}_{N',N} &=&  A^{(.,*)}_{N,N'} \frac{\xi_{N}^*}{\xi_{N'}^*} \quad \mbox{for $N'=N+1,N,N-1$}\ .
\eeqa

The Leonard pair $(A,A^*)$ associated with the AW relations (\ref{aw1}), (\ref{aw2}) is known to be of {\it $q$-Racah type}: the eigenvalue sequences are of the form \cite[Theorem 4.4 (case I)]{Ter03}:
\beqa
\theta_M= \tb q^{2M} + \tc q^{-2M}\ , \quad   \theta^*_N= \tb^* q^{2N} + \tc^*q^{-2N}\ ,\label{st}
\eeqa
where $\tb,\tc,\tb^*,\tc^* \in {\mathbb C}^*$ are such that multiplicities in the spectra do not occur. More details are given in Subsection \ref{sec2}.

\subsubsection{Off-shell and on-shell Bethe states}\label{sec:onoffBs}
In the algebraic Bethe ansatz setting, the starting point of the construction of Bethe states associated with the Leonard pair $(A,A^*)$ is the identification of so-called reference states. Let $m_0$ be an integer. 
By \cite[Propositions 3.1, 3.2]{BP19}, if the parameter $\alpha$ is such that:
	\beqa
	\mbox{$(q^2-q^{-2})\chi^{-1}\alpha \tc^*q^{m_0}=1$ \qquad (resp. $(q^2-q^{-2})\chi^{-1}\alpha \tb q^{-m_0}=-1$)}\label{ab}
	\eeqa
	then 
	\beqa
	\pi(\mathscr{C}^+(u,m_0))
	|\theta^*_0\rangle =0\, \qquad \mbox{(resp.  $\pi(\mathscr{C}^-(u,m_0))
		|\theta_0\rangle  =0\,$)}.\label{cmO}
	\eeqa

	A correspondence between the references states $|\Omega^\pm\rangle$ and the fundamental eigenvectors for Leonard pairs  follows \cite[Definition 3.1]{BP22}:
\beqa
&&|\theta_0\rangle = |\Omega^-\rangle  \ , \quad   |\theta^{*}_0\rangle = |\Omega^+\rangle 
 \ .\label{defrefs}
\eeqa
More generally, different types of  Bethe states are built from successive actions of the dynamical operators $\mathscr{B}^\pm(u,m)$
on each reference state
$|\Omega^\pm\rangle$. 
 Namely, consider the string of dynamical operators 
\ben
B^{\epsilon}(\bar u,m,M)&=&\mathscr{B}^{\epsilon}(u_1,m+2(M-1))\cdots \mathscr{B}^{\epsilon}(u_M,m)\, \label{SB}
\een
where we denote $ \bar  u$ the set of variables $\bar  u = \{u_1,u_2,\dots,u_M\}$. 
\begin{defn}\label{def:Bs} The Bethe states generated from the Leonard pair $(A,A^*)$ 
	 are ($m_0\in {\mathbb N}$):
	\ben\label{PsiA}
	|\Psi_{-}^M( \bar u,m_0)\rangle = \pi(B^{-}( \bar  u,m_0,M))|\Omega^{-}\rangle\,\quad
	\mbox{for} \quad (q^2-q^{-2})\chi^{-1}\alpha \tb q^{-m_0}=-1 \quad \mbox{and}\quad \beta=0\ ,
	\een
	\ben\label{PsiAp}
	|\Psi_{+}^M( \bar  w,m_0)\rangle = \pi(B^{+}( \bar  w,m_0,M))|\Omega^{+}\rangle\,\quad
	\mbox{for} \quad (q^2-q^{-2})\chi^{-1}\alpha \tc^*q^{m_0}=1 \quad \mbox{and}\quad \beta=0\ .
	\een
\end{defn}	
As usual, if the set of variables $\bar u$ or $\bar w$ 
is arbitrary, the Bethe states (\ref{PsiA}), (\ref{PsiAp}) are called `off-shell'. If the set of variables $\bar u$ or $\bar w$
 satisfies certain Bethe ansatz equations, the Bethe states  are called `on-shell'.\vspace{1mm}
 \begin{rem} The Bethe states of Definition \ref{def:Bs} associated with $\beta=0$ are suitable for studying integrable models with an Hamiltonian of the form $A$ (or $A^*$) or $\kappa A + \kappa^*A^*$. See e.g. \cite[Section 5.2]{BP19}. To handle an Hamiltonian of the so-called Heun-Askey-Wilson form $\kappa A + \kappa^*A^* + \kappa_+[A,A^*]_q + \kappa_-[A^*,A]_q$, the case $\beta\neq 0$ has to be considered instead.
 \end{rem}

A correspondence  between on-shell Bethe states (associated with either homogeneous Bethe equations or inhomogeneous Bethe equations) and eigenvectors of Leonard pairs is derived in \cite{BP22}. We refer the reader to this work for more details.
Recall the functions $f(u,v)$ and $h(u,v)$ in Appendix \ref{apA} and introduce the notations ($\epsilon=\pm 1$, $\zeta\in{\mathbb C}^*$):
\ben\label{Lap1}
&&\Lambda_1^\epsilon(u)=\frac{q^{-2s-1}}{u^{\epsilon}}
\left( q^{2 s+1} u\zeta^{-1}-u^{-1}\zeta\right)\left( q^{2 s+1}u\zeta-u^{-1}\zeta^{-1}\right) 
\\ && \qquad\qquad\qquad
\times
\left(u \tc^* q^{-2s}  +u^{-1} \tb q^{2s}\right)
\left(
u \left(\frac{\tc}{\tc^*}\right)^{\frac{1-\epsilon}{2}} +
u^{-1} \left(\frac{\tb^*}{\tb}\right)^{\frac{1+\epsilon}{2}}
\right)\,,\nonumber
\een
\ben
&& \Lambda_2^\epsilon(u)=
\frac{(u^2-u^{-2})q^{-2s-1}}{u^\epsilon (qu^2-q^{-1}u^{-2})}
\left( q^{2 s-1} u^{-1}\zeta-u \zeta^{-1}\right)
\left(q^{2 s-1} u^{-1} \zeta^{-1}-u \zeta\right) 
\label{Lap2} \\ && \qquad\qquad\qquad
\times
\left(q^2 u \tb q^{2s}+u^{-1} \tc^* q^{-2s}\right)
\left(q^2 u \left(\frac{\tb^*}{\tb}\right)^{\frac{1+\epsilon}{2}}+u^{-1} 
\left(\frac{\tc}{\tc^*}\right)^{\frac{1-\epsilon}{2}}\right)\, .\nonumber
\een 
Note that $\Lambda_i^\epsilon(u)$, $i=1,2$, are the eigenvalues associated with the dynamical operators $\mathscr{A}^{\epsilon}(u,m),\mathscr{D}^{\epsilon}(u,m)$ acting on the reference state $|\Omega^\epsilon \rangle$ \cite[Lemma 3.3]{BP19}, respectively. As will be detailed  in Subsections \ref{sec2}, \ref{sec3}, the parameter $\zeta$ finds a natural interpretation  with respect to the parameter sequences associated with a Leonard triple. \vspace{3mm}

$\bullet$ \underline{On-shell Bethe states of homogeneous type:}\vspace{2mm}

Define the set of functions:
\ben
E_{\pm}^M(u_i, \bar u_i)=-\frac{b(u_i^2)}{b(qu_i^2)}\prod_{j=1,j\neq i}^Mf(u_i,u_j)\Lambda_1^\pm(u_i)+\prod_{j=1,j\neq i}^Mh(u_i,u_j)\Lambda_2^\pm(u_i)\,,\label{Bfunc}
\een
for $i=1,\dots,M$,  where we denote ${\bar u}_i=\bar u\backslash u_i$. For each integer $M$ (resp. $N$) with $0\leq M,N\leq 2s$, assume there exists at least one set of non trivial admissible\footnote{Note that the set of equations $E_{\pm}^M(u_i, \bar  u_i)=0$
	for $i=1,\dots,M$
	contain \textit{trivial} solutions where $u_i^2=u_i^{-2}$ or $u_i=0$ for $i=1,\dots,M$,
	or solutions such that $U_i=U_j$ where $U_i=(qu_i^2+q^{-1}u_i^{-2})/(q+q^{-1})$. Such solutions are {\it not admissible}.   For further informations, see \cite{BP19,BP22} and \cite[Subsection 2.3]{S22}.} Bethe roots $S^{M(h)}_-=\{u_1,...,u_{M}\}$ (resp. $S_+^{*N(h)}=\{w_1,...,w_{N}\}$) such that
\beqa
E_{-}^M(u_i, \bar u_i)=0 \quad \mbox{for} \quad  \bar u = S^{M(h)}_- \ ,\qquad
(\mbox{resp.}\quad  E_{+}^N(u_i, \bar  u_i)=0 \quad \mbox{for} \quad \bar  u = S^{*N(h)}_+)\ .\label{BAEeqhom}
\eeqa
\vspace{1mm}

\begin{defn}\label{def:homoBA}
	The equations $E_{-}^M(u_i, \bar  u_i)=0$ (resp. $E_{+}^N(u_i, \bar  u_i)=0$) for $i=1,\dots,M$ (resp. $i=1,...,N$)
	are called the Bethe ansatz equations of {\it homogeneous} type associated with the set of Bethe roots ${\bar u}=S^{M(h)}_-$ (resp. ${\bar u}= S^{*N(h)}_+$). 
\end{defn}

As shown in \cite[Lemma 3.5]{BP22} the eigenvectors of $A$ (resp. $A^*$) can be written as on-shell Bethe states of homogeneous type:
\beqa
|\theta_M\rangle &=& {\cal N}_M(\bar u) |\Psi_{-}^M(\bar u,m_0)\rangle \qquad \mbox{for}\quad  \bar u = S^{M(h)}_-\ ,\label{norm1}\\
|\theta_N^{*}\rangle &=& {\cal N}_N^*(\bar w) |\Psi_{+}^N(\bar w,m_0)\rangle \qquad \mbox{for}  \quad  \bar w = S^{*N(h)}_+\ 
\label{norm2}
\eeqa
with normalization factors ${\cal N}_0(.)={\cal N}_0^*(.)=1$ and
\beqa
{\cal N}_M(\bar u)= \prod_{k=1}^M\left(q u_k b(u_k^2)A^{(*,.)}_{k,k-1}\right)^{-1} \ ,\qquad
{\cal N}_N^*(\bar w)= \prod_{k=1}^N\left(-q^{-1} w_k^{-1} b(w_k^2)A^{(.,*)}_{k,k-1}\right)^{-1}   \ ,\label{Ncoeff}
\eeqa
where (\ref{tridAstar}), (\ref{tridAstar2}).
\vspace{2mm}

\newpage

$\bullet$ \underline{On-shell Bethe states of inhomogeneous type:}\vspace{2mm}

Define the functions:
\ben
E_{\pm}(u_i,\bar u_i)&=& \frac{b(u_i^2)}{b(qu_i^2)}u_i^{\pm 1}\prod_{j=1,j\neq i}^{2s}f(u_i,u_j)\Lambda_1^{\pm}(u_i)
-
(q^2u_i^3)^{\mp 1}\prod_{j=1,j\neq i}^{2s}h(u_i,u_j)\Lambda_2^{\pm}(u_i)
\label{Bfuncinhom}\\
&&
\  +    \nu_\pm
\frac{u_i^{\mp 2}b(u_i^2)}{b(q)}\frac{\prod_{k=0}^{2s}b(q^{1/2+k-s}\zeta u_i)b(q^{1/2+k-s}\zeta^{-1}u_i)}{\prod_{j=1,j\neq i}^{2s}b(u_iu_j^{-1}) b(qu_iu_j)}\,\ ,\nonumber
\,
\een
where 
\ben
\nu_+ = q^{-1-4s}\tc^*\,,\quad
\nu_- = q^{1+4s}\tb\,,
\label{deltad}
\een
for $i=1,\dots,2s$. 
For each integer $M$ (or $N$) with $0\leq M,N\leq 2s$,  there exists at least one set of non trivial admissible Bethe roots $S_+^{M(i)}=\{u_1,...,u_{2s}\}$ (resp. $S_-^{*N(i)}=\{w_1,...,w_{2s}\}$) such that 
\beqa
E_{+}(u_j,\bar u_j)=0 \quad \mbox{for} \quad \bar u = S^{M(i)}_+ 
\ \qquad
(\mbox{resp.}\quad  E_{-}(w_j,\bar w_j)=0 
\quad \mbox{for} \quad \bar w = S^{*N(i)}_-)
\ ,\label{BAEeqinhom}
\eeqa
and 
\beqa
\theta_M&=&  q^{-4s} \Big( \tc^* (\zeta^2+\zeta^{-2})[2s]_q +
q^{2s}(\tb q^{2s}+\tc q^{-2s})
-q\tc^*\sum_{j=1}^{2s} (qu_j^2+q^{-1}u_j^{-2})\Big) \quad \mbox{for} \quad \bar u = S^{M(i)}_+\ \label{eigd1}\\
\qquad \quad  \mbox{(resp.} \ \  \theta_N^{*}&=& q^{4s}\Big( \tb  (\zeta^2+\zeta^{-2})[2s]_q +
q^{-2s}(\tb^* q^{2s}+\tc^* q^{-2s})
-q^{-1}\tb\sum_{j=1}^{2s} (qw_j^2+q^{-1}w_j^{-2}) \Big) \ \quad \label{eigd2}\\ \nonumber&\mbox{for}& \quad \bar w = S^{*N(i)}_-\ )\ . 
\eeqa
\begin{defn}\label{def:inhomoBA}
	The equations $E_{+}(u_i, \bar  u_i)=0$ (resp. $E_{-}(u_i, \bar  u_i)=0$) for $i=1,\dots,2s$
are called the Bethe ansatz equations of {\it inhomogeneous} type associated with the set of Bethe roots ${\bar u}=S^{M(i)}_+$ (resp. ${\bar u}= S^{*N(i)}_-$). 
\end{defn}

As shown in \cite[Lemma 3.7]{BP22} the eigenvectors of $A$ (resp. $A^*$) can be written as on-shell Bethe states of inhomogeneous type:
\beqa
|\theta_M\rangle &=& {\cal N}^{(i)}_M(\bar u') |\Psi_{+}^{2s}(\bar u',m_0)\rangle \qquad \mbox{for}\quad  {\bar u'}= S_+^{M(i)}\ ,\label{normi1}\\
|\theta_N^{*}\rangle &=& {\cal N}^{*(i)}_N(\bar w') |\Psi_{-}^{2s}(\bar w',m_0)\rangle \qquad \mbox{for}  \quad  {\bar w'}=  S^{*N(i)}_-\ 
\label{normi2}
\eeqa
with
\beqa
{\cal N}^{(i)}_M(\bar u')= {\cal N}_{2s}^{*}(\bar u')({P^{\{.,*\}}}^{-1})_{2s,M} \ ,\qquad
{\cal N}^{*(i)}_N(\bar w')=  {\cal N}_{2s}(\bar w')P^{\{.,*\}}_{2s,N} \ ,\label{Ncoeffi}
\eeqa
where (\ref{Ncoeff}) and (\ref{PMN}) are used.

\section{Scalar products of Bethe states from the Leonard triple}
In this section, the theory of Leonard triples that extends the theory of Leonard pairs is applied to the derivation of certain scalar products of off-shell Bethe states (\ref{PsiA}), (\ref{PsiAp}) with eigenvectors $|\theta_M\rangle$. To this end, the Leonard triple $(A,A^*,A^{\diamond})$ is built from the Leonard pair $(A,A^*)$, and the transition matrices relating the respective eigenbases associated with the triple are given. Using the properties of $A^\diamond$, two different types of normalized scalar products of Bethe states are computed. Specializing those scalar products, we obtain an explicit relation characterizing Bethe roots of inhomogenous type in terms of Bethe roots of homogeneous type. 

\subsection{The Leonard triple $(A,A^*,A^{\diamond})$}\label{sec2}
The concept of Leonard triple is a natural generalization of the concept of 
Leonard pair. According to \cite[Definition 1.2]{Curt07}, a Leonard triple on $\cV$ is a 3-tuple of maps
in ${\rm End}(\cV)$ such that for each map, there exists a basis of $\cV$ with respect to which the
matrix representing that map is diagonal and the matrices representing the other two
maps are irreducible tridiagonal.

A Leonard triple of $q$-Racah type $(A,A^*,A^{\diamond})$ is constructed as follows. Recall the Leonard pair of $q$-Racah type $(A,A^*)$ with spectra (\ref{st}) introduced in Subsection \ref{sec:LP}. Denote 
by ${\mathbb I}$ the $(2s+1) \times (2s+1)$ identity matrix. Introduce the notations:
\beqa
\theta^a_M= \tb^a q^{2M} + \tc^a q^{-2M}\ , \quad M=0,1,...,2s,  \qquad a\in\{. ,*,\diamond\}\ \label{thetatrip}
\eeqa
with $\tb^a,\tc^a \in {\mathbb C}^*$ and
\beqa
\omega^{\{a,b,c\}} &=& -(q-q^{-1})^2 \left(  \theta^a_{s}\theta^b_{s} -\frac{1}{r_0} (q^{2s+1} + q^{-2s-1})\theta^c_{s} \right)\ \quad \mbox{with\ \ \ \  $r_0^{-2}=\tb\tc=\tb^*\tc^*=\tb^\diamond\tc^\diamond$}\ .\label{omegatrip}
\eeqa
 Define
\beqa
A^\diamond = \frac{r_0}{(q^2-q^{-2})}\big[ A^*,A\big]_q + \frac{r_0\ \omega^{\{.,*,\diamond\}} }{(q-q^{-1})(q^2-q^{-2})}{\mathbb I}\label{Adiamond}\ . \label{Adiam}
\eeqa
By \cite[Theorem 1.5]{T04} the Leonard pair $(A,A^*)$	satisfies the image by $\pi$ of the Askey-Wilson relations (\ref{aw1})-(\ref{aw2}) for some structure constants.
The proof of the following Lemma is straightforward, thus we skip it. 
\begin{lem} Assume that the structure constants in  (\ref{aw1})-(\ref{aw2}) take the form:
	\beqa
	\rho= -\frac{(q^2-q^{-2})^2}{r_0^{2}}\ ,\quad \omega = \omega^{\{.,*,\diamond\}} \ ,\quad \eta = \frac{(q+q^{-1})}{r_0}\omega^{\{.,\diamond,*\}} \ ,\quad \eta^* = \frac{(q+q^{-1})}{r_0}\omega^{\{*,\diamond,.\}}\ .\label{structc}
	\eeqa
	The triple $(A,A^*,A^{\diamond})$ satisfy  the Askey-Wilson relations:
\beqa
\qquad && \big[A^a,\big[A^a,A^b\big]_q\big]_{q^{-1}}=  -\frac{(q^2-q^{-2})^2}{r_0^{2}}\,A^b+ \omega^{\{a,b,c\}} A^a+\frac{(q+q^{-1})}{r_0}\omega^{\{a,c,b\}}{\mathbb I} \quad \mbox{for any} \quad a\neq b \in \{.,*,\diamond\}\ . \label{AWtriple} 
\eeqa
\end{lem}
\begin{rem}  The Askey-Wilson relations (\ref{AWtriple}) for $(a,b,c)\in\{(.,*,\diamond),(*,.,\diamond)\}$ together with the relation (\ref{Adiam}) can be alternatively put into the ${\mathbb Z}_3$-symmetric form:
	\beqa
	\qquad \frac{r_0}{(q^2-q^{-2})}\big[ A^b,A^a\big]_q - A^c + \frac{r_0\ \omega^{\{ a,b,c\}} }{(q-q^{-1})(q^2-q^{-2})}=0\label{Apairabc}\quad \mbox{for $(a,b,c)\in\{(. ,*,\diamond), (*,\diamond,. ), (\diamond,. ,*)\}$}\ .\label{Z3AW}
	\eeqa
Note that the Askey-Wilson relations considered in \cite[eqs. (27)-(29)]{HuangLT} are obtained from the following  specialization of  \eqref{Z3AW}:
		\beqa
	&& A|_{q\rightarrow q^{-1}} \rightarrow  A \ , \quad A^*|_{q\rightarrow q^{-1}}	 \rightarrow  A^* \ , \quad A^\diamond|_{q\rightarrow q^{-1}} \rightarrow A^\varepsilon\ ,\quad 2s \rightarrow d, \quad r_0 \rightarrow 1 \ , \label{specH}\\
	&& \tb \rightarrow a^{-1}q^{2s}\ , \quad  \tc \rightarrow aq^{-2s}\ ,\quad \tb^* \rightarrow b^{-1}q^{2s}\ , \quad  \tc^* \rightarrow bq^{-2s}\ ,\quad
		 \tb^\diamond \rightarrow -c^{-1}q^{2s}\ , \quad  \tc^\diamond \rightarrow -cq^{-2s}\ .\nonumber
		 	\eeqa
\end{rem}
\begin{rem} The relationship between the algebraic Bethe ansatz parameter $\zeta$ introduced in Subsection \ref{sec22}  (see also \cite{BP19,BP22}) and the parameters $\tb^\diamond,\tc^\diamond$ in \eqref{thetatrip} is determined through a comparison between the structure constants  \cite[eqs. (2.12)-(2.14)]{BP22} and the structure constants in (\ref{AWtriple}) for $(a,b,c)=(.,*,\diamond)$. One finds:
\beqa
\zeta^{-2} =\tb^\diamond r_0 q^{2s} \quad \mbox{and} \quad \zeta^{2} =\tc^\diamond r_0 q^{-2s} \ .\label{bcdiamzeta}
\eeqa
\end{rem}
\vspace{1mm}

\subsubsection{Eigenbases and respective actions of the Leonard triple} We now turn to the explicit construction of three eigenbases $\{|\theta_M^a \rangle \}_{M=0}^{2s}, \{|\theta_M^b \rangle \}_{M=0}^{2s} , \{|\theta_M^c \rangle \}_{M=0}^{2s}$  and respective actions associated with the triple $(A^a,A^b,A^c)$ for the sets $(a,b,c)\in \{ (. ,*,\diamond), (* ,\diamond, .), (\diamond ,.,*)\}$.
\begin{lem}\label{lem:Aab}
	Assume $(A^a,A^b)$ is a Leonard pair of $q$-Racah type with spectra   (\ref{thetatrip}). There exists  eigenbases  $\{|\theta_M^a \rangle \}_{M=0}^{2s}$ and $\{|\theta^b_M\rangle \}_{M=0}^{2s}$ such that:
	\beqa
	\qquad \quad A^a |\theta^a_M\rangle &=& \theta^a_M|\theta^a_M\rangle \ ,\quad A^b |\theta^a_M\rangle =
	A^{(b,a)}_{M+1,M} |\theta^a_{M+1}\rangle +  A^{(b,a)}_{M,M} |\theta^a_{M}\rangle +  A^{(b,a)}_{M-1,M}|\theta^a_{M-1}\rangle \ ,\label{tridAa} \\
	A^b |\theta^b_M\rangle &=& \theta^b_M|\theta^b_M\rangle \ ,\quad
	A^a |\theta^b_M\rangle =
	A^{(a,b)}_{M+1,M} |\theta^b_{M+1}\rangle +  A^{(a,b)}_{M,M} |\theta^b_{M}\rangle +  A^{(a,b)}_{M-1,M}|\theta^b_{M-1}\rangle \ \label{tridAb}
	\eeqa
	with 
	\beqa
	A^{(b,a)}_{M,M-1}&=&   q^{2-4s} \frac{ (1-q^{2M})(\tc^a-\tb^a q^{2M+4s}) (\tb^b\tb^cr_0q^{4s-1} +\tb^a q^{2M-2})(   \tc^a \tc^c r_0 q^{-1} + \tc^b q^{2M-2}) }{(\tc^a-\tb^a q^{4M-2})(\tc^a-\tb^a q^{4M})}  \ ,\label{ammtrip1}\\
	A^{(b,a)}_{M-1,M}&=&  \frac{ (1-q^{2M-4s-2})(\tc^a- \tb^a q^{2M-2}) (\tc^a+\tb^b\tb^c r_0 q^{2M+4s-1})  (\tc^b+\tb^a\tc^c r_0 q^{2M-1})}{(\tc^a-\tb^a q^{4M-4})(\tc^a-\tb^a q^{4M-2})} \ ,\label{ammtrip2}\\
	A^{(b,a)}_{M,M}&=& \theta^b_0 -  A^{(b,a)}_{M,M+1} - A^{(b,a)}_{M,M-1} \ ,\label{ammtrip3} \ 
	\eeqa
	and
	\beqa
	A^{(a,b)}_{M,N} = A^{(b,a)}_{M,N}|_{\tb^a \leftrightarrow \tb^b, \tc^a \leftrightarrow \tc^b, \tb^c q^{4s} \leftrightarrow \tc^c} \ .\label{AMNab}
	\eeqa
\end{lem} 

\begin{proof}
	We adapt the results of \cite{Ter04} to our conventions: in \cite[eqs. (134)-(137)]{Ter04}, firstly  we substitute $q \rightarrow q^2$, $\theta_0 \rightarrow h(1+sq^2),  \theta^*_0 \rightarrow h^*(1+s^*q^2)$ and replace
	\beqa
	&& h \rightarrow \tc^a\ ,\quad h^*\rightarrow \tc^b\ ,\quad d \rightarrow 2s\ , \quad s \rightarrow \frac{\tb^a}{\tc^a}q^{-2}\ ,\quad s^* \rightarrow \frac{\tb^b}{\tc^b}q^{-2}\ , \label{conv1}\\
	&& r_1   \rightarrow -\frac{\tb^b\tb^c}{\tc^a}r_0 q^{4s-1}\ ,\quad
	r_2   \rightarrow -\frac{\tb^a\tc^c}{\tc^b}r_0 q^{-1}\ 
	\ . \label{conv2}
	\eeqa
Then, according to the above substitutions, expressions for the entries of $A^a$ (resp. $A^b$) in the eigenbasis of $A^b$ (resp. $A^a$) are computed using \cite[bases $d^*0^*0d$ and $d00^*d^*$ page 23]{Ter04}. 
\end{proof}

We now compute the action of $A^c$ as defined in (\ref{Apairabc}) on the eigenbases  $\{|\theta_M^a \rangle \}_{M=0}^{2s}$ and $\{|\theta^b_M\rangle \}_{M=0}^{2s}$ such that (\ref{tridAa}), (\ref{tridAb}) hold. Below, the proof of the first lemma is given. The second lemma is proven along the same line.
\begin{lem}\label{lem:f} Assume that $(A^c,A^a)$ is a Leonard pair of $q$-Racah type  with spectra   (\ref{thetatrip}). Then
		\beqa
	A^c |\theta^a_M\rangle =
	A^{(c,a)}_{M+1,M}\frac{f_{M+1}}{f_M} |\theta^a_{M+1}\rangle +  {A}^{(c,a)}_{M,M} |\theta^a_{M}\rangle +  A^{(c,a)}_{M-1,M}\frac{f_{M-1}}{f_M}|\theta^a_{M-1}\rangle \ ,\label{Acact}
	\eeqa
	where
\beqa
f_M = f_0 \prod_{k=0}^{M-1} \mathfrak{f}(k)\,,\quad \mathfrak{f}(k) = \frac{\tc^a q^{-2k-2}+r_0 q^{-1}\tc^b\tb^c}
{\tb^aq^{2k}+r_0q^{-1}\tc^b\tb^c}\ \label{fM}
\eeqa
	for some scalar $f_0\in {\mathbb C}^*$. One sets $f_{-1}=0$.
	\end{lem}
\begin{proof} On one hand, using (\ref{Apairabc}) we find:
	 	\beqa
	 A^c |\theta^a_M\rangle =
	 \hat{A}^{(c,a)}_{M+1,M} |\theta^a_{M+1}\rangle +  \hat{A}^{(c,a)}_{M,M} |\theta^a_{M}\rangle +  \hat{A}^{(c,a)}_{M-1,M}|\theta^a_{M-1}\rangle \ ,\label{Acactbis}
	 \eeqa
	 where
	 \beqa
	 \hat{A}^{(c,a)}_{M+1,M}=q^{-2M-1}r_0\tc^a A^{(b,a)}_{M+1,M}\,,
	 \quad
	 \hat{A}^{(c,a)}_{M-1,M}=q^{2M-1}r_0\tb^a A^{(b,a)}_{M-1,M}\,,
	 \eeqa
	 \beqa
	 \hat{A}^{(c,a)}_{M,M}=\frac{r_0\omega^{\{a,b\}}}{(q-q^{-1})(q^2-q^{-2})}
	 +\frac{r_0}{q+q^{-1}}\theta^a_M A^{(b,a)}_{M,M}\ .
	 \eeqa
	On the other hand, since $(A^c,A^a)$ is a Leonard pair, there exists a basis of $A^a$ with eigenvectors  $|\bar\theta^a_M\rangle = f_M |\theta^a_M\rangle$ such that:
	\beqa
	\qquad \quad A^a |\bar\theta^a_M\rangle &=& \theta^a_M|\bar\theta^a_M\rangle \ ,
	\quad A^c |\bar\theta^a_M\rangle =
	A^{(c,a)}_{M+1,M} |\bar\theta^a_{M+1}\rangle +  A^{(c,a)}_{M,M} |\bar\theta^a_{M}\rangle +  A^{(c,a)}_{M-1,M}|\bar\theta^a_{M-1}\rangle \ .\label{tridAac}
	\eeqa
	Comparing the second equation in
	(\ref{tridAac}) to (\ref{Acact}), after simplifications one gets the desired coefficients. In particular,  the identity relating the diagonal coefficients in (\ref{Acact}) and (\ref{tridAac}) leads to the non-trivial relation
	\beqa
	\frac{r_0\omega^{\{a,b,c\}}}{(q-q^{-1})(q^2-q^{-2})}
	+\frac{r_0}{q+q^{-1}}\theta^a_M\left(\theta^b_0 -  A^{(b,a)}_{M,M+1} - A^{(b,a)}_{M,M-1}\right)=
	\theta^c_0 -  A^{(c,a)}_{M,M+1} - A^{(c,a)}_{M,M-1}\, ,
	\eeqa
	which is readily checked.
\end{proof}
\begin{lem}\label{lem:g} Assume that $(A^b,A^c)$ is a Leonard pair of $q$-Racah type  with spectra   (\ref{thetatrip}). Then
	\beqa
	A^c |\theta^b_M\rangle =
	A^{(c,b)}_{M+1,M}\frac{g_{M+1}}{g_M} |\theta^b_{M+1}\rangle +  {A}^{(c,b)}_{M,M} |\theta^b_{M}\rangle +  A^{(c,b)}_{M-1,M}\frac{g_{M-1}}{g_M}|\theta^b_{M-1}\rangle \ ,\label{Acact1}
	\eeqa
	where
	\beqa
	g_M = g_0 \prod_{k=0}^{M-1} \mathfrak{g}(k)\,,\quad
	\mathfrak{g}(k) = q^{4k}\frac{\tb^b}{\tc^b}\frac{\tc^bq^{-2k}+r_0 q \tc^a\tb^c}
	{\tb^bq^{2k}+r_0q^{-1}\tc^a\tb^c}  \label{gM}
	\eeqa
	for some scalar $g_0\in {\mathbb C}^*$. One sets $g_{-1}=0$.
\end{lem}

We now introduce the eigenbasis $\{|\theta_M^c \rangle \}_{M=0}^{2s}$ of $A^c$ and give  the action of $A^b$ and $A^a$ on its vectors. Adapting the results in \cite[eqs. (134)-(137)]{Ter04} to our conventions, one gets the following lemma.
\begin{lem} Assume that $(A^b,A^c)$ is a Leonard pair of $q$-Racah type. There exists an eigenbasis $\{|\theta_M^c \rangle \}_{M=0}^{2s}$ of $A^c$ such that:
\beqa
A^b |\theta^c_M\rangle =
A^{(b,c)}_{M+1,M} |\theta^c_{M+1}\rangle +  A^{(b,c)}_{M,M} |\theta^c_{M}\rangle +  A^{(b,c)}_{M-1,M}|\theta^c_{M-1}\rangle \ ,\label{tridAc}
\eeqa
where $A^{(b,c)}_{M,N}$ is given by (\ref{AMNab}) for $(a,b,c)\rightarrow (b,c,a)$.
\end{lem}
The previous lemma fixes the normalizations of the eigenvectors $|\theta^c_{M}\rangle $. For the proof of the following lemma, we use (\ref{Apairabc}) with the substitution $(a,b,c)\rightarrow (b,c,a)$ and (\ref{tridAc}), and follow the same reasoning as for Lemmas \ref{lem:f}, \ref{lem:g}.
\begin{lem}\label{lem:h} Assume that $(A^c,A^a)$ is a Leonard pair of $q$-Racah type  with spectra   (\ref{thetatrip}). Then
	\beqa
	A^a |\theta^c_M\rangle =
	A^{(a,c)}_{M+1,M}\frac{h_{M+1}}{h_M} |\theta^c_{M+1}\rangle +  {A}^{(a,c)}_{M,M} |\theta^c_{M}\rangle +  A^{(a,c)}_{M-1,M}\frac{h_{M-1}}{h_M}|\theta^c_{M-1}\rangle \ ,\label{Aaact}
	\eeqa
	where
		\beqa
h_M = h_0 \prod_{k=0}^{M-1} \mathfrak{h}(k)\,,\quad \mathfrak{h}(k) = 
q^{4k}\frac{\tb^c}{\tc^c}\frac{\tc^cq^{-2k}+r_0 q \tc^b\tb^a}
{\tb^cq^{2k}+r_0q^{-1}\tc^b\tb^a} \label{hM}
	\eeqa
	for some scalar $h_0\in {\mathbb C}^*$. One sets $h_{-1}=0$.
\end{lem}
\begin{rem} The structure constants in (\ref{AWtriple}) are recovered by applying \cite[Theorem 4.5 and 5.3]{T04} for the spectra (\ref{thetatrip}) and diagonal entries (\ref{ammtrip3}).
\end{rem}

Up to now, we have assumed that $(A^a,A^b)$, $(A^b,A^c)$ and $(A^c,A^a)$ are Leonard pairs. In order to show that $(A^a,A^b,A^c)$  is a Leonard triple, it remains to identify the conditions such that those assumptions hold.
\begin{prop}  If the following conditions
	\beqa
	&& {\rm{(i)}} \quad \ \tb^a(\tc^a)^{-1}\neq q^{-2M}\ ,\quad
	\tb^b(\tc^b)^{-1}\neq q^{-2M}\ ,\quad
	\tb^c(\tc^c)^{-1}\neq q^{-2M}\ , \quad \mbox{for}\quad 1 \leq M \leq   4s-1\ ,\nonumber\\
	&&{\rm{(ii)}}\nonumber \quad
	\tc^b\neq -r_0q^{-2M+1} \tc^a\tc^c,
	\quad
	\tb^a\neq -r_0q^{-2M+1+4s} \tb^b\tb^c,
	\quad
	\tc^a \neq -r_0q^{2M+4s-1}  \tb^b\tb^c,
	\quad
	\tc^b \neq -r_0q^{2M-1} \tb^a \tc^c,
	\quad
	\\ &&  \nonumber \qquad \quad\quad\mbox{for}\quad 1 \leq M \leq   2s \ ,
	\eeqa
	are satisfied, then $(A^a,A^b,A^c)$ is a Leonard triple of $q$-Racah type.
\end{prop}
\begin{proof}
	Firstly, we identify the conditions (i) that ensure the spectra (\ref{thetatrip}) have no multiplicities. From (\ref{thetatrip}), one has:
	\beqa
	\theta_M^a - \theta_N^a = (q^{M-N}-q^{-M+N})(\tb^a q^{M+N} - \tc^a q^{-M-N})\ .\nonumber
	\eeqa
	Recall $q$ is not a root of unity. Observe that this latter equation is vanishing only if the conditions (i) are not satisfied.  Now assume (i). Secondly, we identify the conditions  (ii) such that for each map (e.g. $A^a$), with respect to its eigenbasis (e.g. $\{|\theta_M^a \rangle \}_{M=0}^{2s}$) the matrices representing the other two maps (e.g. $A^b,A^c$)
	are irreducible tridiagonal. In other words, we identify the relations between the parameters such that off-diagonal upper and lower entries entering in  (\ref{tridAa}), (\ref{tridAb}), (\ref{Acact1}), (\ref{tridAc}) and  (\ref{Acact}), (\ref{Aaact}) are not vanishing. Assuming (i), (ii) are satisfied, we conclude $(A^a,A^b,A^c)$ is a Leonard triple.
\end{proof}
\begin{rem}
The conditions (i), (ii) specialize to \cite[Definition 7.1 (RQRAC3), (RQRAC4)]{HuangLT} 	using \eqref{specH}. 
\end{rem}

\subsubsection{Transition matrices} 
From the theory of Leonard pairs, it is known that the transition matrix relating the two different eigenbases can be expressed in terms of orthogonal  polynomials of the Askey-scheme \cite{Ter04}. 
In previous subsection, three different eigenbases associated with a Leonard triple $(A^a,A^b,A^c)$ have been introduced. Below, the transition matrices relating the three eigenbases are given explicitly in terms of $q$-Racah polynomials. For convenience, we denote: 
\beqa
P^{\{a,b\}}_{MN}=k_N^{\{a,b,c\}} R^{\{a,c\}}_M(\theta^b_N) \qquad \mbox{and}  \qquad ({P^{\{a,b\}}}^{-1})_{NM} =   {\nu^{\{a,b,c\}}_0}^{-1} k^{\{b,a,c\}}_M R^{\{a,c\}}_M(\theta^b_N)\ \label{PMNtrip}
\eeqa
and
\beqa
R^{\{a,c\}}_M(\theta^b_N)= \fpt{  q^{-2M}  }{  \frac{\tb^a}{\tc^a}q^{2M}  }{  q^{-2N}  }{  \frac{\tb^b}{\tc^b}q^{2N}  }{  -\frac{\tb^a\tc^c}{\tc^b}r_0q}{  -\frac{\tb^b\tb^c}{\tc^a}r_0q^{4s+1}}{  q^{-4s}  }\ , \label{qracahpoly}
\eeqa
\beqa
k^{\{a,b,c\}}_N= \frac{ (-\frac{\tb^b\tb^c}{\tc^a}r_0 q^{4s+1},-\frac{\tb^a\tc^c}{\tc^b}r_0 q,\frac{\tb^b}{\tc^b},q^{-4s};q^2)_N }{ ( q^2, -\frac{\tb^b\tb^c}{\tb^a}r_0q,  -\frac{\tc^a\tc^c}{\tc^b}r_0q^{1-4s}, \frac{\tb^b}{\tc^b}q^{4s+2} ;q^2)_N} \frac{(1-\frac{\tb^b}{\tc^b}q^{4N})}{(\frac{\tb^a}{\tc^a})^N(1-\frac{\tb^b}{\tc^b})}      \ ,\qquad k^{\{b,a,c\}}_M= k_M|_{\tb^a\leftrightarrow \tb^b,\tc^a\leftrightarrow \tc^b, \tb^bq^{4s} \leftrightarrow \tc^c} \ , \nonumber
\eeqa
\beqa
\nu^{\{a,b,c\}}_0=   \frac{ (\frac{\tb^a}{\tc^a}q^2, \frac{\tb^b}{\tc^b}q^{2} ;q^2)_{2s}}{(-\frac{\tb^b\tb^c}{\tc^a}r_0q^{4s+1})^{2s} ( -\frac{\tb^a\tc^c}{\tb^b}r_0q^{1-4s},  -\frac{\tc^a\tc^c}{\tc^b}r_0q^{1-4s};q^2)_{2s}}\ .\nonumber  
\eeqa

\begin{rem} The correspondence between the variables entering  in the standard definition of the $q-$Racah polynomial $R_n(\mu(x);\alpha,\beta,\gamma,\delta|q)$ where $\mu(x)=q^{-x} +\gamma\delta q^{x+1}$ in \cite[eq. (3.2.1)]{KS96}, and the Bethe ansatz parameters $\tb^a,\tc^a,s$, $a\in\{.,*,\diamond\}$, entering in the Bethe ansatz equations \eqref{BAEeqhom} and \eqref{BAEeqinhom} with \eqref{eigd1}, \eqref{eigd2} is given by:
\beqa
&& q \rightarrow q^2\ ,\quad	n \rightarrow M\ ,\quad	x \rightarrow N\ ,\nonumber\\
&&	\alpha \rightarrow -\frac{\tb^a\tc^c}{\tc^b}r_0q^{-1}\ ,\quad \beta \rightarrow -\frac{\tb^c\tc^b}{\tc^a}r_0q^{-1}\ ,\quad \gamma\rightarrow q^{-4s-2}\ ,\quad \delta \rightarrow \frac{\tb^b}{\tc^b}q^{4s}\ . \nonumber
\eeqa
\end{rem}

\begin{prop}  The transition matrix relating the basis $\{|\theta_N^b\rangle\}_{N=0}^{2s}$ and $\{|\theta_M^a\rangle\}_{M=0}^{2s}$
	is such that:
	\beqa
	|\theta^b_N\rangle = \sum_{M=0}^{2s} P^{\{a,b\}}_{MN}  |\theta^a_M \rangle \qquad \mbox{and} \qquad |\theta^a_M\rangle = \sum_{N=0}^{2s}
	({P^{\{a,b\}}}^{-1})_{NM}|\theta^b_N \rangle \ ;\label{transeqtrip}
	\eeqa 
	The transition matrix relating the basis $\{|\theta_N^c\rangle\}_{N=0}^{2s}$ and $\{|\theta_M^b\rangle\}_{M=0}^{2s}$
	is such that:
	\beqa\label{ctob}
	|\theta_N^c\rangle = \sum_{M=0}^{2s} g_MP_{MN}^{\{b,c\}}|\theta_M^b\rangle
	\qquad \mbox{and} \qquad
	|\theta^b_M\rangle = \sum_{N=0}^{2s} 
	g_M^{-1}({P^{\{b,c\}}}^{-1})_{NM}|\theta^c_N \rangle \ ;
	\eeqa
	The transition matrix relating the basis $\{|\theta_N^a\rangle\}_{N=0}^{2s}$ and $\{|\theta_M^c\rangle\}_{M=0}^{2s}$
	is such that:
	\beqa\label{ctoa}
|\theta_N^a\rangle = \sum_{M=0}^{2s} 	f_N^{-1}h_MP_{MN}^{\{c,a\}}|\theta_M^c\rangle
	\qquad \mbox{and} \qquad
	|\theta^c_M\rangle = \sum_{N=0}^{2s} 
	f_Nh_M^{-1}({P^{\{c,a\}}}^{-1})_{NM}|\theta^a_N \rangle \ ,
	\eeqa
	with  (\ref{fM}), (\ref{gM}), (\ref{hM}) and (\ref{PMNtrip}).
\end{prop}
\begin{proof} Firstly, consider (\ref{transeqtrip}). With respect to the coefficients (\ref{tridAa}), (\ref{tridAb}),
the explicit expressions for the transition matrix with entries $P_{MN}^{\{a,b\}}$ - as well as the inverse transfer matrix' entries - are obtained from the results in \cite[Section 16]{Ter04} adapted to our conventions (\ref{conv1}), (\ref{conv2}). Starting from  (\ref{Acact1}), (\ref{tridAc}) and (\ref{Acact}), (\ref{Aaact}),  eqs. (\ref{ctob}) and (\ref{ctoa}) are derived similarly.
\end{proof}

According to Lemmas \ref{lem:f}, \ref{lem:g} and \ref{lem:h}, the three eigenbases associated with the Leonard triple $(A^a, A^b, A^c)$ are defined up to overall scalars $f_0,g_0,h_0$. Their precise relationship follow from consistency of the eigenbases' construction, as we now show.

\begin{prop}
	The scalars $f_0,g_0,h_0$ are such that:
	\beqa
	f_0h_0^{-1}g_0^{-1}=(-1)^{2s}q^{2 s(2s-1)} \frac{\left(q^2 r_0^2 b^2;q^2\right){}_{2 s} \left(q^2 r_0^2 \left(b^{\diamond }\right)^2;q^2\right){}_{2 s} \left(q^2 r_0^2 \left(b^*\right)^2;q^2\right){}_{2 s}}{\left(-\frac{ q r_0 b b^{\diamond }}{b^*};q^2\right){}_{2 s} \left(-\frac{q r_0 b b^*}{b^{\diamond }};q^2\right){}_{2 s} \left(-\frac{q r_0 b^{\diamond } b^*}{b};q^2\right){}_{2 s}}\ .\label{fgh0}
	\eeqa
\end{prop}
\begin{proof}
	Use $|\theta_M^b\rangle$ given by  (\ref{transeqtrip}) in (\ref{ctob}), and compare it to (\ref{ctoa}). One finds: 	\beqa
		f_0h_0^{-1}g_0^{-1}=\frac{\prod_{k=0}^{N-1} \mathfrak{h}(k)\sum_{M=0}^{2s}P_{M'M}^{\{a,b\}}P_{MN}^{\{b,c\}}\prod_{k=0}^{M-1} \mathfrak{g}(k)
		}{({P^{\{c,a\}}}^{-1})_{M'N}\prod_{k=0}^{M'-1} \mathfrak{f}(k)}
		\eeqa
		for all $M'=0,\dots,2s$ and $N=0,\dots,2s$. The l.h.s. being independent of $M',N$, we set $M'=N=0$ in the r.h.s. From (\ref{PMNtrip}), one has $P^{\{a,b\}}_{0M}=k_M^{\{a,b,c\}}$, $P^{\{b,c\}}_{M0}= 1$ and $ ({P^{\{c,a\}}}^{-1})_{00}= {\nu_0^{\{c,a,b\}}}^{-1}$. It follows:
		\beqa
		f_0h_0^{-1}g_0^{-1}= \nu_0^{\{c,a,b\}}\sum_{M=0}^{2s}k_M^{\{a,b,c\}}\prod_{k=0}^{M-1}\mathfrak{g}(k) \ .\label{prod0}
		\eeqa
		Then, using the relation 
		\beqa
		(a;q^{-2})_M = (-a)^{M}q^{-M(M-1)}(a^{-1};q^2)_M\ 
		\eeqa
		and $r_0^{-2}=\tb^a\tc^a =\tb^b\tc^b=\tb^c\tc^c$,
		we obtain: 
		\beqa
		\prod_{k=0}^{M-1}\mathfrak{g}(k) = \frac{(-\frac{qr_0\tb^b\tb^c}{\tb^a};q^2)_M}{(-\frac{qr_0\tb^a\tc^c}{\tc^b};q^2)_M}q^{M(M-1)}\left(\frac{q^{-1}r_0\tb^c\tc^a}{\tb^b}\right)^{-M}\ .\label{prodgk}
		\eeqa
		Moreover, 
		$\nu_0^{\{c,a,b\}}$ simplifies to:
		\beqa
		\nu_0^{\{c,a,b\}} =  \frac{\left(q^2 r_0^2 (\tb^a)^2, q^2 r_0^2 (\tb^c)^2   ;q^2\right){}_{2 s} }{  \left(- \frac{qr_0\tb^a\tb^b}{\tb^c}, - \frac{qr_0\tb^b\tb^c}{\tb^a} ;q^2\right){}_{2 s}  }\frac{1}{\left( - \frac{qr_0\tb^a\tb^c}{\tb^b}\right)^{2s}}\ .
		\eeqa
		Using $k_M^{\{a,b,c\}}$ below \eqref{qracahpoly} and \eqref{prodgk}, one gets:
		\beqa
		\sum_{M=0}^{2s}k_M^{\{a,b,c\}}\prod_{k=0}^{M-1}\mathfrak{g}(k) =   \sum_{M=0}^{2s}\frac{ (-\frac{qr_0\tb^a\tb^c}{\tc^b}q^{4s},\frac{\tb^b}{\tc^b},q^{-4s};q^2)_M }{ (   -\frac{\tb^b}{qr_0\tb^a\tb^c}q^{-4s+2}, \frac{\tb^b}{\tc^b}q^{4s+2} , q^2;q^2)_M} \frac{(1-\frac{\tb^b}{\tc^b}q^{4M})}{(\frac{q^{-1}r_0\tb^a\tb^c}{\tb^b})^{M} (1-\frac{\tb^b}{\tc^b})} q^{M(M-1)} \ .
		\eeqa
		In order to show \eqref{fgh0}, it remains to prove  that the r.h.s. of \eqref{fgh0} equals the r.h.s. of \eqref{prod0}. This is equivalent to the equality
		\beqa
		\sum_{M=0}^{2s}\frac{ (-\frac{qr_0\tb^a\tb^c}{\tc^b}q^{4s},\frac{\tb^b}{\tc^b},q^{-4s};q^2)_M }{ (   -\frac{\tb^b}{qr_0\tb^a\tb^c}q^{-4s+2}, \frac{\tb^b}{\tc^b}q^{4s+2} , q^2;q^2)_M} \frac{(1-\frac{\tb^b}{\tc^b}q^{4M})}{ (1-\frac{\tb^b}{\tc^b})(\frac{q^{-1}r_0\tb^a\tb^c}{\tb^b})^{M}} q^{M(M-1)} =
		\frac{\left(q^2 r_0^2 (\tb^b)^2;q^2\right){}_{2 s} }{\left(-\frac{ q^{1-4s} \tb^b}{r_0\tb^a\tb^c};q^2\right){}_{2 s}} \ .\label{relorth}
		\eeqa
		This latter relation is actually a limiting case and specialization of the orthogonality relation satisfied by the $q$-Racah polynomials, that is obtained as follows. Set $n=m=0$ and substitute $q\rightarrow q^2$ in \cite[eq. (3.2.2)]{KS96} and consider its limit $\beta \rightarrow \infty$. Then, substitute
		\beqa
		N\rightarrow 2s\ ,\quad x\rightarrow M\ ,\quad \alpha\rightarrow q^{-4s-2} \ ,\quad \gamma\rightarrow - \frac{\tb^a\tb^c}{\tc^b}r_0q^{4s-1} \ ,\quad \delta\rightarrow - \frac{\tb^b}{r_0\tb^a\tb^c}q^{-4s-1}\ ,
		\eeqa
		in the resulting equation, we find \eqref{relorth}.	
\end{proof}
In Appendix \ref{apB}, an explicit example of a Leonard triple is given for $s=1/2$. 

\subsection{Computation of scalar products}\label{sec3}
The representation theory of Leonard triples of $q$-Racah type  is now applied to the computation of scalar products of any eigenvector of $A$ with certain off-shell Bethe states. 
Recall the definition of the off-shell Bethe states \eqref{PsiA}, \eqref{PsiAp} and the relation between the parameter $\zeta$ and the parameters $\tb^\diamond,\tc^\diamond$ given by (\ref{bcdiamzeta}). For convenience, let us introduce the notation:
\beqa
U_i&=& \frac{qu_i^2+ q^{-1}u_i^{-2}}{q+q^{-1}}\label{Uisym} \ .
\eeqa 

\begin{lem} For $\beta=0$, one has:
	\beqa
\pi( 	B^{\epsilon}(\bar u,m,M))&=& {\cal G}_M^{\epsilon}(\bar{u})  \prod_{i=1}^M \left( U_i - \frac{r_0}{q+q^{-1}}A^\diamond \right) \label{stringB}
	\eeqa
	with
	\beqa
	{\cal G}_M^{\epsilon}(\bar{u}) = \left(\epsilon(q+q^{-1})q^{-\epsilon(M+1)} \tb^{\frac{1-\epsilon}{2}}{\tc^*}^{\frac{1+\epsilon}{2}}\right)^M  \left( \prod_{i=1}^M b(u_i^2)u_i^{-\epsilon}\label{G} \right)
	\eeqa
\end{lem}
\begin{proof} For $\beta=0$, (\ref{Bm}) simplifies to
	\beqa
	&&\pi( \mathscr{B}^{\epsilon}(u,m))=b(u^2)u^{-\epsilon}\frac{q^{-\epsilon(m+2)}}{\alpha}\Big( A^\diamond - \frac{(q+q^{-1})}{r_0}U\Big)\ 
	\eeqa
	with the identification (\ref{Adiamond}). Then, the image of the string of dynamical operators $\pi( B^{\epsilon}(\bar u,m,M))$ with \eqref{SB} is reduced, and depends on $\alpha$ and $m$. Now, consider the constraint on $\alpha$ according to the definition of the reference state. Choosing  $\chi=-(q^2-q^{-2})/r_0$ in  (\ref{PsiA}) (resp. (\ref{PsiAp}))  one has $\alpha=q^m/r_0\tb$ for $\epsilon=-$  (resp. $\alpha=-q^{-m}/r_0\tc^*$ for $\epsilon=+$). In each case $\epsilon=\pm$, one replaces the corresponding value of $\alpha$ in the reduced expression of $\pi( 	B^{\epsilon}(\bar u,m,M))$. After simplifications, one gets \eqref{stringB}
\end{proof}
The off-shell Bethe states \eqref{PsiA}, \eqref{PsiAp} are now expressed in the eigenbasis of $A$.
\begin{lem}
	\beqa
	|\Psi_{-}^{M''}( \bar  u,m )\rangle  &=&  {\cal G}_{M''}^{-}(\bar{u})
	\sum_{N=0}^{2s} \sum_{M'=0}^{2s}\frac{f_N}{f_0} ({P^{\{\diamond,.\}}}^{-1})_{NM'} \prod_{i=1}^{M''} \left( U_i - \frac{r_0}{q+q^{-1}}\theta_{M'}^\diamond \right) |\theta_{N}\rangle
	\label{BmoinsOmn} \\
	|\Psi_{+}^{M''}( \bar  u,m )\rangle &=&  {\cal G}_{M''}^{+}(\bar{u})
	\sum_{N=0}^{2s} \sum_{M'=0}^{2s} \frac{ f_{N}}{{\nu_0^{\{*,\diamond, . \}}}g_0h_{M'}} ({P^{\{\diamond,. \}}}^{-1})_{NM'}
	\prod_{i=1}^{M''} \left( U_i - \frac{r_0}{q+q^{-1}}\theta_{M'}^\diamond \right) 
	|\theta_{N}\rangle
	\ .\label{BplusOmp} 
	\eeqa
\end{lem}
\begin{proof} To show (\ref{BmoinsOmn}), we proceed as follows. Consider the l.h.s. of (\ref{BmoinsOmn}) defined by (\ref{PsiA}) with (\ref{defrefs}). Using the first relation in (\ref{ctoa}) for $N=0,M\rightarrow M'$ and $(c,a,b)\rightarrow (\diamond,.,*)$, the reference state $|\theta_0\rangle$ decomposes into the  eigenvectors $|\theta^\diamond_{M'}\rangle$ of $A^\diamond$. Note that $P^{\{\diamond,. \}}_{M'0}=1$. Apply (\ref{stringB}) to this combination. Then, rewrite the resulting expression in terms of the eigenvectors of $A$ using the second relation in (\ref{ctoa}). It reduces to the r.h.s. of (\ref{BmoinsOmn}).
	 The proof of (\ref{BplusOmp}) is done similarly, observing that $({P^{\{*,\diamond \}}}^{-1})_{M'0}={\nu_0^{\{*,\diamond,.  \}}}^{-1}$ is independent of $M'$.
\end{proof}

We are now ready to give the main result of this paper, the normalized scalar products of on-shell versus off-shell Bethe states, using the properties of the Leonard triple and related eigenbases.  Recall \eqref{G}, \eqref{fgh0}, \eqref{fM}, \eqref{hM}, \eqref{thetatrip}, \eqref{PMNtrip}.
\begin{thm}\label{prop:scalprod}
	\beqa
	\frac{\langle \theta_M	|\Psi_{-}^{M''}( \bar u,m)\rangle}{\langle \theta_M	| \theta_M \rangle} &=&
	{\cal G}_{M''}^{-}(\bar{u})\frac{f_M}{f_0} 
	\sum_{M'=0}^{2s}({P^{\{\diamond,. \}}}^{-1})_{MM'} \prod_{i=1}^{M''} \left( U_i - \frac{r_0}{q+q^{-1}}\theta_{M'}^\diamond \right)
	\ ,\label{scal1} \\
	\frac{	\langle \theta_M	|\Psi_{+}^{M''}( \bar  u,m)\rangle}{\langle \theta_M	| \theta_M \rangle} &=& 
	{\cal G}_{M''}^{+}(\bar{u}){\frac{f_{M}}{g_0\nu_0^{\{*,\diamond,.\}}}}
	\sum_{M'=0}^{2s} \frac{1}{h_{M'}} ({P^{\{\diamond,.\}}}^{-1})_{MM'}
	\prod_{i=1}^{M''} \left( U_i - \frac{r_0}{q+q^{-1}}\theta_{M'}^\diamond \right) 
	\ .\label{scal2}
	\eeqa
\end{thm}
\begin{proof} To show (\ref{scal1}), we apply $\langle \theta_M	|$ to (\ref{BmoinsOmn}) and use (\ref{orthcond}).  We show  (\ref{scal2}) similarly.
\end{proof}

In view of the next section, let us   display the results for one of  the simplest case, $s=1/2$ and  $M''(=2s)=1$. For the two examples below, we denote:
\beqa
X^{\pm,off}_{M}(u_1) = 	\frac{\langle \theta_M	|\Psi_{\pm}^{1}( u_1 ,m)\rangle}{\langle \theta_M	| \theta_M \rangle} \ , \quad M=0,1\ .
\eeqa
Recall the notation (\ref{b}) and the identification (\ref{bcdiamzeta}).
\begin{example}\label{ex1}
	Consider $s=1/2$ and $\epsilon=-$. From   (\ref{scal1}) we obtain:
	\beqa
	{X}_{0}^{-,off}(u_1)&=& \frac{ q^3 r_0^2 \tb^2}{b(q)\left(q^2 r_0^2 \tb^2;q^2\right){}_{1}}
	u_1b(u_1^2)Y_-(u_1|u_1)\,\ ,\label{X0m}\\
	{X}_{1}^{-,off}(u_1)&=&- 
	\frac{ q b(q) \tb\left(-\frac{q r_0 \tb \tb^{\diamond }}{\tb^*};q^2\right){}_{1}
		\left(-\frac{q r_0 \tb^{\diamond } \tb^*}{\tb};q^2\right){}_{1}}
	{r_0 \tb^{\diamond } \left(q^2 r_0^2 \tb^2;q^2\right){}_{1}}u_1b(u_1^2) \,,\label{X1m}
	\eeqa
	with 
	\beqa
	Y_-(x|x)=\frac{b(q)}{q^3r_0^3\tb \tb^*\tb^\diamond}
	\Bigg(
	\frac{q r_0 \tb^{\diamond } \tb^* \left( q^2 r_0^2\tb^2-1\right)}{x^2}+q^3 x^2 r_0 \tb^{\diamond } \tb^* \left( q^2 r_0^2 \tb^2-1\right)+\tb^* \left(1- q^4 r_0^2\tb^2\right)
	\\  - \left(q^2-1\right) q r_0\tb \tb^{\diamond }
	\left(q^2 r_0^2 \left(\tb^*\right)^2+1\right)+q^2 r_0^2 \left(\tb^{\diamond }\right)^2 \tb^* \left(1- q^4 r_0^2 \tb^2\right) \Bigg)\ .\nonumber
	\eeqa
\end{example}

\begin{example}\label{ex2}
	Consider $s=1/2$ and $\epsilon=+$. From (\ref{scal2}) we obtain:
	\beqa
	{X}_{0}^{+,off}(u_1)&=&\frac{1}{q^4 r_0^3 \tb^*\tb^\diamond}
	u_1^{-1}b(u_1^2)Z_1(u_1)
	\,\ ,\label{X0p}\\
	{X}_{1}^{+,off}(u_1)&=& \frac{1}{q^4r_0^3\tb \tb^*\tb^\diamond}
	u_1^{-1}b(u_1^2)Z_2(u_1)\ ,\label{X1p}
	\,
	\eeqa
	where $Z_i(x)$ are Laurent polynomials in $x$ given by,
	\beqa
	Z_1(x)=q^3  r_0 \tb^{\diamond }x^2+\frac{q r_0 \tb^{\diamond }}{x^2}-1+q^2 r_0^2 \tb^{\diamond } \left( q \left(q^2-1\right) r_0 \tb \tb^*-\tb^{\diamond }\right)\,,\label{Z1}
	\eeqa
	\beqa
	Z_2(x)= q^3  r_0 \tb \tb^{\diamond }x^2+\frac{ q r_0 \tb \tb^{\diamond }}{x^2}+
	q \left(q^2-1\right) r_0 \tb^{\diamond } \tb^*-\tb \left(q^2 r_0^2 \left(\tb^{\diamond }\right)^2+1\right).	\label{Z2}
	\eeqa
\end{example}

	\begin{rem}
		As an example of application of Theorem \ref{prop:scalprod} to an elementary integrable system,  on one hand consider the Hamiltonian $H(\lambda)$ obtained from specializing \cite[eq. (5.2)]{BP19} for $\kappa_\pm=0$ and identifying $\lambda=\frac{\kappa^*}{\kappa}$:
	\beqa
	H(\lambda) &=& S_1^xS_2^x + S_1^yS_2^y +\Delta S_1^zS_2^z + \lambda(S_2^xS_3^x + S_2^yS_3^y +\Delta S_2^zS_3^z) \nonumber\\ &&\qquad +\frac{1}{q-q^{-1}}( S_1^z + (\lambda-1)S_2^z - \lambda S_3^z)\ ,\qquad \Delta=\frac{q+q^{-1}}{2}\ ,\nonumber
	\eeqa
where 
\beqa\label{Pauli}
S^x = \frac{1}{2}\left(
\begin{array}{cc}
	0    & 1 \\
	1 & 0 
\end{array} \right)\ ,\quad S^y=\frac{1}{2}\left(
\begin{array}{cc}
	0    & -i \\
	i & 0 
\end{array} \right)  \ , \quad S^z = \frac{1}{2}\left(
\begin{array}{cc}
	1    & 0 \\
	0 & -1 
\end{array} \right)\ .
\eeqa	
This Hamiltonian is the image in ${\rm End}({\mathbb C}^2\otimes {\mathbb C}^2 \otimes {\mathbb C}^2)$ of $A + \lambda A^*$ via the specialization of the map \cite[eq. (2.14), (2.15)]{BP19}.
We refer the reader to \cite[Subsection 5.2.1]{BP19} where the solution (spectrum and eigenstates) is discussed in details. Using the framework of Leonard pairs, the Bethe eigenstates of $H(\lambda) $ are given by \eqref{PsiA} (or \eqref{PsiAp}) for $M=2s$  where $\bar u$ satisfy the inhomogeneous Bethe equations \cite[eq. (3.75)]{BP19} for $\epsilon=-$ (or $\epsilon=+$). On the other hand, consider the Hamiltonian $H'= H(0)$. 
	The eigenstates are given by \eqref{norm1}.
	Thus, the ratios in Theorem \ref{prop:scalprod} for $\bar u$ satisfying \cite[eq. (3.75)]{BP19} determine the overlap between the eigenstates of $H(\lambda) $ and $H'$. In the context of quantum phase transitions where $\lambda$ is viewed as the driving parameter (see e.g. \cite[Section III]{Gu08}), those ratios determine the fidelity  when considering the ground states of each Hamiltonian. 
	\end{rem}

\vspace{1mm}

\subsection{Relating inhomogeneous and homogeneous Bethe roots}
In Subsection \ref{sec:onoffBs}, we have seen that the eigenvector $|\theta^*_N \rangle$ of $A^*$  can be either written as the on-shell Bethe state (\ref{norm2}) associated with the Bethe equations of homogeneous type $E_+^{N}(u_j,\bar u_j)=0$ with (\ref{Bfunc}), $j=0,...,N$, or as the on-shell Bethe state (\ref{normi2}) associated with the Bethe equations of inhomogeneous type $E_-(u_j,\bar u_j)=0$, $j=0,...,2s$ with (\ref{Bfuncinhom}), (\ref{eigd2}). In this section,
a set of equations relating the inhomogenous and homogeneous Bethe roots associated with the eigenvalue $\theta^*_N$, $N=0,....,2s$, is obtained. 

Let us denote the symmetric Bethe roots
 for the homogeneous case and inhomogeneous case as follows:
\beqa
U_j^{(i)} 
\quad \mbox{for} \quad \bar u =  S^{*N(i)}_- \ ,\quad
U_j^{(h)}
\quad \mbox{for} \quad \bar u =  S^{*N(h)}_+ \ \quad \mbox{with (\ref{Uisym})} .\label{symBR}
\eeqa 
We also introduce the notation
$\bar U_j = \{U_1, U_2, ..., U_{j-1}, U_{i+1}, ..., U_N\}$.
From \cite[Proposition 3.5]{BP19}, each system of Bethe equations is equivalent to a system of polynomial equations:
\beqa
E_+^{N}(u_j,\bar u_j)=0 \quad \Leftrightarrow \quad P^{N}_+(U^{(h)}_j,\bar U^{(h)}_j) &=& 0 \ \quad \mbox{for} \quad j=1,...,N\ ,\label{polyh}\\
E_-(u_j,\bar u_j)=0 \quad \Leftrightarrow \quad P^{2s}_-(U^{(i)}_j,\bar U^{(i)}_j) &=& 0 \ \quad \mbox{for} \quad j=1,...,2s\ .\label{polyi}
\eeqa
respectively. Here, $P^{N}_+(U_j,\bar U_j)$ (resp.  $P^{2s}_-(U_j,\bar U_j)$) has maximal degree $N$ (resp. degree $2s+1$) in the variables $U_1,...,U_N$ (resp. $U_1,...,U_{2s}$). For more details, see \cite[Subsection 3.5.1]{BP19}. 
Following \cite[Conjectures 2,3]{BP19} based on numerical supporting evidences:
\begin{hyp}\label{hyp1} The system of polynomial equations (\ref{polyh}) admits a unique admissible solution \\$\bar U^{(h)}= \{U^{(h)}_1, U^{(h)}_2, ..., U^{(h)}_N \}$ up to permutation.
\end{hyp} 
\begin{hyp}\label{hyp2} The system of polynomial equations (\ref{polyi}) admits $2s+1$ distinct admissible solution \\
	$\bar U^{(i)}= \{U^{(i)}_1, U^{(i)}_2, ..., U^{(i)}_{2s} \}$ up to permutation.
\end{hyp} 

In order to relate the two systems of Bethe roots, we specialize Theorem \ref{prop:scalprod}. The following lemmas give two relations between $q$-Racah polynomials associated with different parameter sequences, that involve either homogeneous Bethe roots or  inhomogenous Bethe roots. 
\begin{lem} Assume Hypothesis \ref{hyp1}. The $q$-Racah polynomial  $R^{\{.,\diamond\}}_M(\theta^*_{N})$ admits the following decomposition of homogeneous type in terms of $R^{\{\diamond ,*\}}_{M'}(\theta_{M})$:
	\beqa
	R^{\{.,\diamond\}}_M(\theta^*_{N})&=&\frac{f_M}{f_0}
	\frac{ \sum_{M'=0}^{2s} \prod_{j=1}^{N} \left( U_j^{(h)} - \frac{r_0}{q+q^{-1}}\theta_{M'}^\diamond \right)\frac{1}{h_{M'}}  k_{M'}^{\{. ,\diamond,*\}} R^{\{\diamond ,*\}}_{M'}(\theta_{M})}{\sum_{M'=0}^{2s}
		\prod_{j=1}^{N} \left( U_j^{(h)} - \frac{r_0}{q+q^{-1}}\theta_{M'}^\diamond \right)\frac{1}{h_{M'}} k_{M'}^{\{.,\diamond,*\}} }
	\quad \mbox{with} \quad {\bar u}=  S^{*N(h)}_+
	\ .\label{Racid2}
	\eeqa
\end{lem}
\begin{proof} Firstly, we compute the normalized scalar product $\langle \theta_M	|\theta_N^* \rangle / \langle \theta_M	|\theta_M \rangle$. To do that, we multiply both sides of the equality in (\ref{scal2}) by 	${\cal N}^{*}_N(\bar u)$ with (\ref{norm2}) and specialize the result to ${\bar u}=  S^{*N(h)}_+\ $. After simplifications using (\ref{PMNtrip}), it follows:
	\beqa
	\frac{\langle \theta_M	|\theta_N^* \rangle }{ \langle \theta_M	|\theta_M \rangle} = {\cal N}^{*}_N(\bar u)
	{\cal G}_{N}^{+}(\bar{u}) {\nu^{\{*,\diamond,. \}}_0}^{-1}\frac{f_M}{g_0} \sum_{M'=0}^{2s} \prod_{j=1}^{N} \left( U_j^{(h)} - \frac{r_0}{q+q^{-1}}\theta_{M'}^\diamond \right)\frac{1}{h_{M'}}
	k_{M'}^{\{.,\diamond,*\}} R^{\{\diamond ,*\}}_{M'}(\theta_{M}) \quad \mbox{for} \quad {\bar u}=  S^{*N(h)}_+
	\ .\nonumber
	\eeqa
	Combining above expression with the fact that the $q$-Racah polynomials are given by (see \cite[Theorem 14.6 and 15.6]{T03}):
	\beqa
	R^{\{.,\diamond\}}_M(\theta^*_N) &=& \frac{\langle \theta_M |\theta_N^* \rangle}{\langle \theta_0 |\theta^*_N \rangle}\frac{\langle \theta_0 |\theta_0 \rangle}{\langle \theta_M |\theta_M \rangle}\ ,\label{Rac1}
	\eeqa
	we get	(\ref{Racid2}). 
\end{proof}
Similar arguments are used to derive the following lemma, thus we skip the proof.
\begin{lem} Assume Hypothesis \ref{hyp2}. The $q$-Racah polynomial  $R^{\{.,\diamond\}}_M(\theta^*_{N})$ admits the following decomposition of inhomogeneous type in terms of $R^{\{\diamond ,*\}}_{M'}(\theta_{M})$:
	\beqa
	R^{\{. ,\diamond\}}_M(\theta^*_{N})&=&\frac{f_M}{f_0}
	\frac{ \sum_{M'=0}^{2s} \prod_{j=1}^{2s} \left( U_j^{(i)} - \frac{r_0}{q+q^{-1}}\theta_{M'}^\diamond \right)  k_{M'}^{\{.,\diamond,*\}} R^{\{\diamond ,*\}}_{M'}(\theta_{M})}{\sum_{M'=0}^{2s} \prod_{j=1}^{2s} \left( U_j^{(i)} - \frac{r_0}{q+q^{-1}}\theta_{M'}^\diamond \right) k_{M'}^{\{.,\diamond,*\}} }
	\quad \mbox{with} \quad {\bar u}=  S^{*N(i)}_-
	\ .\label{Racid1}
	\eeqa
\end{lem}

Combining above lemmas, it follows:
\begin{prop}\label{relinhhom}  Assume Hypothesis 1 and  2. Let $U_j^{(h)},j=1,...,N$ be the symmetric Bethe roots associated with the eigenvalue $\theta_N^*$ given in (\ref{st}) and homogeneous Bethe equations $E_+^N(u_j,\bar u_j)=0$ with (\ref{Bfunc}). The inhomogenous symmetric Bethe roots $U_j^{(i)},j=0,...,2s$ associated with $\theta_N^*$ satisfy the relations:
	\beqa
	\frac{ \prod_{j=1}^{2s} \left( U_j^{(i)} - \frac{r_0}{q+q^{-1}}\theta_{M'}^\diamond \right) }{\sum_{M'=0}^{2s} \prod_{j=1}^{2s} \left( U_j^{(i)} - \frac{r_0}{q+q^{-1}}\theta_{M'}^\diamond \right) k_{M'}^{\{.,\diamond,*\}} } = \frac{\frac{1}{h_{M'}} \prod_{j=1}^{N} \left( U_j^{(h)} - \frac{r_0}{q+q^{-1}}\theta_{M'}^\diamond \right) }{\sum_{M'=0}^{2s} \frac{1}{h_{M'}}\prod_{j=1}^{N} \left( U_j^{(h)} - \frac{r_0}{q+q^{-1}}\theta_{M'}^\diamond \right) k_{M'}^{\{.,\diamond,*\}} }\label{BAih}\  
	\ ,\eeqa	
 for $M'=0,1,...,2s$.\end{prop}
\begin{proof} The difference between the r.h.s. of (\ref{Racid1}) and the r.h.s. of (\ref{Racid2}) is vanishing. Using the fact that the $q$-Racah polynomials $R^{\{\diamond ,*\}}_{M'}(\theta_{M})$ are linearly independent,  one gets  (\ref{BAih}).
\end{proof}
\begin{example}\label{example3} Consider $s=1/2$. From (\ref{BAih}), we obtain\footnote{Note that for $N=0$ there are no homogeneous Bethe roots in (\ref{BAih}).}:
	\beqa
	\mbox{For $N=0$:} && U_1^{(i)}=\left(\frac{r_0}{q+q^{-1}}\right)
	\left(\frac{h_0 \theta_0^\diamond-h_1 \theta_1^\diamond}{h_0-h_1}\right)\,,\label{U1i0}\\
		\mbox{For $N=1$:} && U_1^{(i)}=\left(\frac{r_0}{q+q^{-1}}\right)
	\left(
	\frac{\left(\frac{r_0}{q+q^{-1}}\right)(h_1-h_0)\theta_0^\diamond\theta_1^\diamond
		+(h_0\theta_0^\diamond-h_1\theta_1^\diamond)U_1^{(h)}}
	{\left(\frac{r_0}{q+q^{-1}}\right)(h_1\theta_0^\diamond-h_0\theta_1^\diamond)
		+(h_0-h_1)U_1^{(h)}}
	\right)\,.\label{U1i1}
	\eeqa
	From \eqref{polyh} with \eqref{Bfunc}, we find a unique admissible
	solution of (\ref{Bfunc},\ref{polyh}) given by
	\beqa
	U_1^{(h)}=\frac{\tb(q^2\tb^*-q^{-2}\tc^*)(1+q^2r_0^2(\tb^\diamond)^2)+(q-q^{-1})r_0^{-1}
		\tb^\diamond(1+q^2r_0^2\tb^2)}
	{(q+q^{-1})r_0\tb\tb^\diamond(q^2\tb^*-\tc^*)}\,,\label{exa3h}
	\eeqa
	 in agreement with Hypothesis \ref{hyp1}. 
Independently, it is checked that $E_{-}(u_i, \bar  u_i)=0$ with (\ref{Bfuncinhom}) admits only two admissible solutions in agreement with  Hypothesis
\ref{hyp2}, and that they coincide with  (\ref{U1i0}) and (\ref{U1i1}) with (\ref{exa3h}).
\end{example}

For higher values of $s=1,3/2$, we proceed as follows. Firstly, admissible homogeneous symmetric Bethe roots $U_j^{(h)},j=1,...,N$,  are determined numerically from \eqref{polyh}, and Hypothesis \ref{hyp1} is checked. Replacing the numerical values in \eqref{BAih}, for each $N=0,...,2s$, the inhomogenous symmetric Bethe roots $U_j^{(i)},j=1,...,2s$, are determined numerically.
Independently, the admissible solutions of $E_{-}(u_i, \bar  u_i)=0$ with (\ref{Bfuncinhom}) are computed, and Hypothesis \ref{hyp2} is checked. Those admissible solutions are found to agree with the solutions previously computed from \eqref{BAih}.
\begin{example} Consider $s=1$ and the numerical values
$q=3$, $r_0=1$, $\tb=5$, $\tb^*=7$ and $\tb^\diamond=1/2$. In Table \ref{tab1} the numerical solutions of the system (\ref{BAih}) for $N=0,1,2$ are given.
\begin{table}[h]
\begin{tabular}{ |c|c|c| } 
 \hline
  & $\{U_1^{(h)},U_2^{(h)}\}$ & $\{U_1^{(i)},U_2^{(i)}\}$ \\ \hline
   $N=0$ & - & $\{3.50405,11.9071\}$ \\ \hline
 $N=1$ & $\{18.5087,-\}$ & $\{3.208,12.0789\}$ \\\hline
  $N=2$ & $\{2.72742,12.0749\}$ & $\{2.01305,12.1539\}$ \\
 \hline
\end{tabular}
\caption{Homogeneous and inhomogeneous Bethe roots that solve \eqref{polyh}, (\ref{BAih}) for the scalars 
$q=3$, $r_0=1$, $\tb=5$, $\tb^*=7$ and $\tb^\diamond=1/2$.}\label{tab1}
\end{table}
\end{example}

\section{Scalar products for $M''=2s$ and Belliard-Slavnov linear system}
 In this section, we consider the normalized scalar products of on-shell and off-shell Bethe states (\ref{scal1}), (\ref{scal2}) at the special value $M''=2s$ from a different perspective. It is shown that those scalar products satisfy  systems of homogeneous linear equations of the form introduced and studied in \cite{BS19}. 
 It implies the existence of a determinant formula for the solutions of Theorem \ref{prop:scalprod} at $M''=2s$, and for $q$-Racah polynomials when those solutions are specialized to Bethe roots of inhomogeneous type. 
\subsection{Belliard-Slavnov linear system associated with a Leonard pair}
We apply the method recently proposed in \cite{BS19}. Below, we denote:
  \beqa
  \bar{\mathcal{Y}}=\{y_1,y_2,...,y_{2s},y_{2s+1}\} \ , \quad \bar{\mathcal{Y}}_j=\bar{\mathcal{Y}}\backslash y_j\ \quad \mbox{and}\quad  \ \bar{\mathcal{Y}}_{j,k} = \bar{\mathcal{Y}}_{j}\backslash y_k\ .
  \eeqa
  Introduce the convenient notations:
  \beqa
  g(u,\bar u) = \prod_{u_i\in \bar u}g(u,u_i) \quad \mbox{with}\quad  g(u,u_i) =\left(b(u/u_i)b(quu_i)\right)^{-1}  \ \label{defgu}
  \eeqa
  and 
  \beqa
  Y_\epsilon(u|\bar u) &=&  \frac{u^{2\epsilon+1} \Lambda_1^\epsilon(u)}{b(u^2)b(qu^2)} \prod_{u_i\in \bar u} b(u/qu_i)b(uu_i)  + \frac{q^{-\epsilon-1}u^{-1} \Lambda_2^\epsilon(u)}{b(u^2)b(q^2u^2)}  \prod_{u_i\in \bar u} b(qu/u_i)b(q^2uu_i) \nonumber\\
  && -  \delta_{\epsilon,+}\nu_+  \prod_{j=0}^{2s}b(q^{1/2+j-s}u\zeta)b(q^{1/2+j-s}u\zeta^{-1})
  \   ,   \label{Yeps}
  \eeqa
  where $\epsilon=\pm$, $\delta_{+,+}=1,\delta_{-,+}=0$ and $\nu_+$ is given in (\ref{deltad}).
  
 Given a Leonard pair $(A,A^*)$,   different systems of  linear equations can be derived, either starting from $A$ or from $A^*$.
 In order to relate with the results of Theorem \ref{prop:scalprod}, here we focus on a system derived from studying the operator $A$ only. As a preliminary, the action of $A$ on two different types of off-shell Bethe states is computed. Recall the structure constants 	$\eta,\eta^*$ given by (\ref{structc}) with (\ref{omegatrip}) and the identification (\ref{bcdiamzeta}).
  \begin{lem} 
  	\beqa
  	A|\Psi_{\epsilon}^{2s}(\bar{\mathcal{Y}}_j,m )\rangle &=&
  	\sum_{k=1}^{2s+1} L^\epsilon_{jk} | \Psi_{\epsilon}^{2s}(\bar{\mathcal{Y}}_{k},m ) \rangle \ ,\quad j=1,...,2s+1 \ ,\label{APsieps}
  	\eeqa
  	with 
  	\beqa
  	L^\epsilon_{jk} &=& \left(\frac{y_j}{y_k}\right)^\epsilon\frac{b(y_k^2)}{b(y_j^2)} g(y_k,\bar{\mathcal{Y}}_{k} ) Y_\epsilon(y_k|\bar{\mathcal{Y}}_{j}) +\frac{(q+q^ {-1})^2}{\rho b(y_j^2)b(q^2y_j^2)}\left(\eta^*+\eta\frac{qy_j^2+q^{-1}y_j^{-2}}{q+q^{-1}}\right)
  	\delta_{jk}
  	\ .\label{Ljkeps}
  	\eeqa
  \end{lem}
  \begin{proof} Firstly, we consider the action of $A$ on off-shell Bethe states. It is obtained by specializing the equations below \cite[eq. (3.76)]{BP19} for $\kappa=1,\kappa^*=0$. One gets the off-shell equation:
  	\ben
  	&&\quad \left(A-\frac{(q+q^ {-1})^2}{\rho b(u^2)b(q^2u^2)}\left(\eta^*+\eta\frac{qu^2+q^{-1}u^{-2}}{q+q^{-1}}\right)\right)|\Psi_{\epsilon}^{2s}(\bar u,m)\rangle
  	\non\\&&=\Lambda_{\epsilon}^{2s}(u,\bar u)|\Psi_{\epsilon}^{2s}(\bar u,m)\rangle
  	+
  	\sum_{k=1}^{2s}\frac{(q-q^{-1})u^\epsilon u_k^{(\epsilon+1)/2} \bar E_{\epsilon}(u_k,\bar u_k)}{(u^2-u^{-2})b(uu_k^{-1})b(quu_k)}|\Psi_{\epsilon}^{2s}(\{u,\bar u_k\},m)\rangle \label{eqAPsi}
  	\een
  	where
  	\ben
  	&&\quad \Lambda_{\epsilon}^{2s}(u,\bar u)=\frac{u^{2\epsilon+1}\Lambda_1^\epsilon(u)}{b(u^2)b(q u^2)}
  	\prod_{j=1}^{2s}f(u,u_j)
  	+\frac{q^{-\epsilon-1}u^{-1}\Lambda_2^\epsilon(u)}{b(u^2)b(q^2u^{2})}
  	\prod_{j=1}^{2s}h(u,u_j)	- \delta_{\epsilon,+} \nu_+
  	\frac{\prod_{j=0}^{2s}b(q^{1/2+j-s}\zeta u)b(q^{1/2+j-s}\zeta^{-1}u)}{\prod_{i=1}^{2s}b(uu_i^{-1})b(quu_i)}\,\nonumber
  	\een
  	and $\bar E_{-}(u_k,\bar u_k)= -E^{2s}_{-}(u_k,\bar u_k)$ with (\ref{Bfunc}), $\bar E_{+}(u_k,\bar u_k)=E_{+}(u_k,\bar u_k)$, with (\ref{Bfuncinhom}).
  	
  	Secondly using the expressions (\ref{defgu}), (\ref{Yeps}), one rewrites:
  	\beqa
  	A|\Psi_{\epsilon}^{2s}(\bar u,m)\rangle&=&  \left(g(u,\bar u) Y_{\epsilon}(u|\bar u) + \frac{(q+q^ {-1})^2}{\rho b(u^2)b(q^2u^2)}\left(\eta^*+\eta\frac{qu^2+q^{-1}u^{-2}}{q+q^{-1}}\right) \right) |\Psi_{\epsilon}^{2s}(\bar u,m)\rangle \nonumber\\
  	&&
  	-\sum_{u_k\in \bar u} \left(\frac{u}{u_k}\right)^{\epsilon}\frac{ b(u_k^2)}{b(u^2)}g(u,u_k)g(u_k,\bar u_k)Y(u_k|\bar u)  |\Psi_{\epsilon}^{2s}(\{u, \bar u_k\},m)\rangle
  	\non\,. 
  	\eeqa
  	Set $u=y_j,\bar u=\bar{\mathcal{Y}}_{j}$. It follows:
  	\beqa
  	A|\Psi_{\epsilon}^{2s}(\bar{\mathcal{Y}}_{j},m)\rangle&=&  \left(g(y_j,\bar{\mathcal{Y}}_{j}) Y_{\epsilon}(y_j|\bar{\mathcal{Y}}_{j}) +\frac{(q+q^ {-1})^2}{\rho b(y_j^2)b(q^2y_j^2)}\left(\eta^*+\eta\frac{qy_j^2+q^{-1}y_j^{-2}}{q+q^{-1}}\right) \right) |\Psi_{\epsilon}^{2s}(\bar{\mathcal{Y}}_{j},m)\rangle \nonumber\\
  	&&
  	-\sum_{k=1}^{2s} \left(\frac{y_j}{y_k}\right)^{\epsilon}\frac{ b(y_k^2)}{b(y_j^2)}\underbrace{g(y_j,y_k)g(y_k,\bar{\mathcal{Y}}_{j,k})}_{=g(y_k,\bar{\mathcal{Y}}_k)}Y_\epsilon(y_k|\bar{\mathcal{Y}}_j)  |\Psi_{\epsilon}^{2s}(\{y_j, \bar{\mathcal{Y}}_{j,k}\},m)\rangle
  	\non\,, 
  	\eeqa
  	which gives (\ref{Ljkeps}).
  \end{proof}
  
  Then, independently of Theorem \ref{prop:scalprod} we have:
  \begin{prop}\label{prop:lineq} For any $M=0,1,...,2s$, the normalized scalar product 
  	\beqa
  	X^{\epsilon,off}_{M}(\bar u)=\frac{\langle \theta_M	|\Psi_{\epsilon}^{2s}( \bar u,m)\rangle}{\langle \theta_M	| \theta_M \rangle}\ \label{spoff}
  	\eeqa
  	solves  the Belliard-Slavnov linear system
  	\beqa
  	\sum_{k=1}^{2s+1} {\cal M}^{\epsilon,M}_{jk}  X^{\epsilon,off}_{ M}( \bar{\mathcal{Y}}_{k}) =0 \quad \mbox{with}\quad {\cal M}^{\epsilon,M}_{jk}= L^\epsilon_{jk} - \delta_{jk}\theta_M
  	\qquad \mbox{for} \quad j=1,2,...,2s+1\ \ \label{lineqXoff}
  	\eeqa
  	with (\ref{Ljkeps}) and  $\bar{\mathcal{Y}}$ are distinct arbitrary complex numbers.
  \end{prop}
  \begin{proof} From (\ref{APsieps}), one has:
  	\beqa
  	\left(A - \theta_M  \mathbb{I} \right) |\Psi_{\epsilon}^{2s}(\bar{\mathcal{Y}}_j,m)\rangle &=&	\sum_{k=1}^{2s+1} \left( L^\epsilon_{jk} - \theta_M \delta_{jk}\right)   | \Psi_{\epsilon}^{2s}(\bar{\mathcal{Y}}_{k},m) \rangle
  	\label{eqM}\\
  	&=&	\sum_{k=1}^{2s+1} {\cal M}^{\epsilon,M}_{jk}  | \Psi_{\epsilon}^{2s}(\bar{\mathcal{Y}}_{k},m) \rangle 	\ .\nonumber
  	\eeqa
  Acting from the left with  $\langle \theta_M|$ and using  (\ref{tridAstardual}), (\ref{orthcond}), one gets (\ref{lineqXoff}). 
  \end{proof}

Independently of Theorem \ref{prop:scalprod}, we may now ask about the existence of non-trivial solutions of the system of homogenous  linear equations (\ref{lineqXoff}).
\begin{lem}\label{lem:rank} Let ${\cal M}^{\epsilon,M}$, $\epsilon=\pm,M=0,...,2s$ denote the matrices with entries ${\cal M}^{\epsilon,M}_{jk}$, $j,k=1,...,2s+1$, given in (\ref{lineqXoff}).
	\beqa
	  {\text rank}({\cal M}^{\epsilon,M})=2s \ .
	 \eeqa
\end{lem}

\begin{proof}
		 The vectors $ | \Psi_{\epsilon}^{2s}(\bar{\mathcal{Y}}_{k},m) \rangle $, $k=1,...,2s+1$, given by \eqref{BmoinsOmn}, \eqref{BplusOmp}  with $M''=2s$ are linearly independent.  By (\ref{eqM}), the matrix with entries  ${\cal M}^{\epsilon,M}_{jk}$ represents the linear map $(\pi(\tA) - \theta_M I) $ with respect to the $2s+1$ vectors $ | \Psi_{\epsilon}^{2s}(\bar{\mathcal{Y}}_{k},m) \rangle $. By linear algebra, $\text{rank}(\pi(\tA) - \theta_M I) = \text{dim} {\cal V} - \text{dim}(\text{kernel} (\pi(\tA)-\theta_M I))$.	By definition, we have $\text{dim} {\cal V}=2s+1$ and one knows that
the dimension of the each eigenspace of $\pi(\tA)$ is 1.
So, $\text{rank}({\cal M}^{\epsilon,M})=\text{rank}(\pi(\tA) - \theta_M I)=2s$. 
\end{proof}
  
The above lemma implies  that the linear system  (\ref{lineqXoff}) admits  non-trivial solutions.
 Results of the previous section provide explicit solutions to (\ref{lineqXoff}). The following is an immediate consequence of Theorem \ref{prop:scalprod} and Proposition \ref{prop:lineq}.
  \begin{thm}\label{thm:BSscal}
The  normalized scalar product of on-shell and off-shell Bethe states given by (\ref{scal1}) (resp. (\ref{scal2})) 
 at $M''=2s$  solves the linear system (\ref{lineqXoff}) at $\epsilon=-$ (resp. $\epsilon=+$).
  \end{thm}
  
  For instance, recall Examples \ref{ex1}, \ref{ex2}. Using the expressions (\ref{X0m}), (\ref{X1m}), (\ref{X0p}), (\ref{X1p}), it is readily checked that the linear system (\ref{spoff})
  	\beqa
  &&{\cal M}^{\epsilon,M}_{11}X^{\epsilon,off}_{M}(y_2) + {\cal M}^{\epsilon,M}_{12}X^{\epsilon,off}_{M}(y_1)=0\non\ ,\\
  &&{\cal M}^{\epsilon,M}_{21}X^{\epsilon,off}_{M}(y_2) + {\cal M}^{\epsilon,M}_{22}X^{\epsilon,off}_{M}(y_1)=0\non \ ,\ 
  \eeqa
  is satisfied for $M=0,1$, $\epsilon=\pm$.  For higher values of $s=1,3/2$, using Mathematica it can be independently checked that
  the expressions (\ref{scal1}) and (\ref{scal2}) for $M''=2s$ solve the linear system (\ref{spoff}) for $\epsilon=-$ and $\epsilon=+$, respectively.

  \subsection{Determinant formula  and $q$-Racah polynomials}
  The fact that the normalized scalar products (\ref{spoff}) solve the systems of  homogeneous linear equations (\ref{lineqXoff}) for $\epsilon=\pm$ implies the existence of a determinant formula for (\ref{scal1}), (\ref{scal2}) at $M''=2s$. 
   Let $T$ be a $2s+1 \times 2s+1$ matrix, and  denote by $T_{[k,\ell]}$ the submatrix of $T$ obtained  by deleting the $k$-th row and  $\ell$-column.
  \begin{prop}\label{prop:scaldet} The normalized scalar product (\ref{spoff}) admits a determinant formula given by:
  		\beqa
  	X^{\epsilon,off}_{M}(\bar{\mathcal{Y}}_{k}) = \psi_M^\epsilon (-1)^{k}\det \widetilde{\cal M}^{\epsilon,M}_{[2s+1,k]}\,,
  	\label{sollin}
  	\eeqa
where  $\widetilde{\cal M}^{\epsilon,M}={\cal W}^{\epsilon,M}{\cal M}^{\epsilon,M}$ for some non-degenerate $2s+1 \times 2s+1$ matrix $\cal W^{\epsilon,M}$ such that:
\beqa
(i) &&\left({\cal W}^{\epsilon,M}{\cal M}^{\epsilon,M}\right)_{2s+1 \ k } = 0 \qquad \mbox{for all} \quad \ k=1,...,2s+1
 \ ,\nonumber\\ (ii) &&\det \widetilde{\cal M}^{\epsilon,M}_{[2s+1,k]} \qquad \mbox{is independent of $y_k$,} \nonumber   
\eeqa 
 and $\psi_M^\epsilon$  some function of $M,\epsilon$.
\end{prop}	
  	\begin{proof} By Lemma \ref{lem:rank}, given $M,\epsilon$ fixed the solutions of the system of homogeneous linear equations (\ref{lineqXoff}) can be written in terms of minors of a certain matrix denoted $\widetilde{\cal M}^{\epsilon,M}$ - up to an overall factor $\psi_M^\epsilon$ -  using Cramer's rule. The matrix $\widetilde{\cal M}^{\epsilon,M}$ is determined as follows.  Firstly, we are looking for a $2s+1 \times 2s+1$ matrix $\cal W^{\epsilon,M}$ such that conditions (i) hold, which turns the system (\ref{lineqXoff}) into the new system:
  		\beqa
  		\sum_{k=1}^{2s} \widetilde{\cal M}^{\epsilon,M}_{jk}  X^{\epsilon,off}_{ M}( \bar{\mathcal{Y}}_{k}) = - \widetilde{\cal M}^{\epsilon,M}_{j2s+1}X^{\epsilon,off}_{ M}( \bar{\mathcal{Y}}_{2s+1}) \quad \mbox{for}\quad  j=1,2,...,2s\ .\ \label{lineqXoffnew}
  		\eeqa
  		Looking for non-trivial solutions using Cramer's rule, we assume $\cal W^{\epsilon,M}$ is non-degenerate such that $\det\widetilde{\cal M}^{\epsilon,M}_{[2s+1,k]} (=\det({\cal W}^{\epsilon,M}{\cal M}^{\epsilon,M})_{[2s+1,k]})\neq 0$. 
  		Moreover, because  the l.h.s. of (\ref{sollin}) is independent of $y_k$  we are looking for a matrix $\cal W^{\epsilon,M}$  such that conditions (ii) hold. 
  		Applying Cramer's rule, the solutions of (\ref{lineqXoffnew}) are written in terms of $X^{\epsilon,off}_{ M}( \bar{\mathcal{Y}}_{2s+1}) $:
  		\beqa
  		X^{\epsilon,off}_{ M}( \bar{\mathcal{Y}}_{k}) = (-1)^{k+2s+ 1} \frac{\det\widetilde{\cal M}^{\epsilon,M}_{[2s+1,k]}}{\det\widetilde{\cal M}^{\epsilon,M}_{[2s+1,2s+1]}} X^{\epsilon,off}_{ M}( \bar{\mathcal{Y}}_{2s+1})\ ,
  		\eeqa
  		which leads to  (\ref{sollin}) for some overall factor $\psi_M^\epsilon$ that is independent of the variables $y_j$, $j=1,...,2s+1$. 
Finally, this factor is determined by comparing  (\ref{scal1}), (\ref{scal2}) for $M''=2s$ at the limit $\bar u \rightarrow \infty$  to the r.h.s. of \eqref{sollin} with \eqref{spoff} at the limit $\bar{\cal Y}_k\rightarrow\infty$.
  	\end{proof}

\vspace{2mm}
  
By specializing the parameters $ \bar{\mathcal{Y}}_{k}$ in (\ref{sollin}) to certain Bethe roots, the off-shell Bethe states entering in the ratio (\ref{spoff}) reduce to on-shell Bethe states that are eigenvectors of the operator $A$. This leads to a determinant formula for $q$-Racah polynomials. Recall  the  parameters' correspondence (\ref{bcdiamzeta}).
\begin{cor}\label{qRacdet} The $q$-Racah polynomial $	R^{\{.,\diamond\}}_M(\theta^*_N) $ with  (\ref{qracahpoly}) admits a determinant formula given by: 
		\beqa
	R^{\{.,\diamond\}}_M(\theta^*_N) 
	 = \frac{\psi_M^-}{\psi_0^-} \frac{\det \widetilde{\cal M}^{-,M}_{[2s+1,2s+1]}}{\det \widetilde{\cal M}^{-,0}_{[2s+1,2s+1]}}\,\qquad \mbox{for\ \quad $\bar{\mathcal{Y}}_{2s+1}=  S^{*N(i)}_-$.}\label{Racdet}
	\eeqa
\end{cor} 
 \begin{proof} Multiply (\ref{spoff}) for $\epsilon=-$ by the second normalization factor in (\ref{Ncoeffi})
 with $\bar u=\bar{\mathcal{Y}}_{2s+1}$. Let us specialize $\bar{\mathcal{Y}}_{2s+1}=  S^{*N(i)}_-$ in the resulting ratio. Using \eqref{normi2}, it follows:
 	\beqa
 	\frac{\langle \theta_M	|\theta_N^* \rangle }{ \langle \theta_M	|\theta_M \rangle} &=& {\cal N}^{*}_N(\bar{\mathcal{Y}}_{2s+1})\frac{\langle \theta_M	|\Psi_{-}^{2s}( \bar{\mathcal{Y}}_{2s+1},m)\rangle}{\langle \theta_M	| \theta_M \rangle}\ \quad \mbox{for} \quad \bar{\mathcal{Y}}_{2s+1}=  S^{*N(i)}_-\ . \label{ratfin}
 	\eeqa
 	Then, we use \eqref{sollin} with (\ref{spoff}) to express the l.h.s. of (\ref{ratfin})  in terms of a determinant. Finally, we use (\ref{Rac1}) to get (\ref{Racdet}).
 \end{proof}
\vspace{1mm}

\subsection{Examples}
To illustrate Proposition  \ref{prop:scaldet} and Corollary \ref{qRacdet}, we consider some examples below.\vspace{2mm}

	  	$\bullet$
	\underline{The case $s=1/2$, $\epsilon=-$:} Consider $M=0$. Firstly, one shows the identity
		\beqa
		g(y_1,y_2)Y_-(y_1|y_2)+\frac{(q+q^ {-1})^2}{\rho b(y_1^2)b(q^2y_1^2)}\left(\eta^*+\eta\frac{qy_1^2+q^{-1}y_1^{-2}}{q+q^{-1}}\right)-\theta_0=g(y_1,y_2)Y_-(y_1|y_1)\,,\non
		\eeqa
		from which we find the entries of $\widetilde{\cal M}^{-,0}$  are given by  
		\beqa
		\widetilde{\cal M}^{-,0}_{11}&=&\frac{g(y_1,y_2)Y_-(y_1|y_1)}{y_2b(y_2^2)}
		\left(y_2b(y_2^2)\cal W^{-,0}_{11}(y_1,y_2)+y_1b(y_1^2)\cal W^{-,0}_{12}(y_1,y_2)\right)\,,\non\\
		\widetilde{\cal M}^{-,0}_{12}&=&\frac{g(y_2,y_1)Y_-(y_2|y_2)}{y_1b(y_1^2)}
		\left(y_2b(y_2^2)\cal W^{-,0}_{11}(y_1,y_2)+y_1b(y_1^2)\cal W^{-,0}_{12}(y_1,y_2)\right)\,,\non\\
		\widetilde{\cal M}^{-,0}_{21}&=&\frac{g(y_1,y_2)Y_-(y_1|y_1)}{y_2b(y_2^2)}
		\left(y_2b(y_2^2)\cal W^{-,0}_{21}(y_1,y_2)+y_1b(y_1^2)\cal W^{-,0}_{22}(y_1,y_2)\right)\,,\non\\
		\widetilde{\cal M}^{-,0}_{22}&=&\frac{g(y_2,y_1)Y_-(y_2|y_2)}{y_1b(y_1^2)}
		\left(y_2b(y_2^2)\cal W^{-,0}_{21}(y_1,y_2)+y_1b(y_1^2)\cal W^{-,0}_{22}(y_1,y_2)\right)\,.\non
		\eeqa
		Then, conditions (i) imply
		\beqa\label{w21}
		\cal W^{-,0}_{21}(y_1,y_2) =-\frac{y_1b(y_1^2)}{y_2b(y_2^2)}\cal W^{-,0}_{22}(y_1,y_2)\ 
		\eeqa
		with $\cal W^{-,0}_{22}(y_1,y_2)$ arbitrary.	
		To satisfy conditions (ii), one imposes $\frac{\partial \widetilde{\cal M}^{-,0}_{12}}{\partial  y_1}=\frac{\partial \widetilde{\cal M}^{-,0}_{11}}{\partial  y_2}=0$. The associated differential equations,
		omitted here for brevity, are solved by
		\beqa\label{w12}
		{\cal W}_{12}^{-,0}(y_1,y_2)= \frac{y_2b(y_2^2)}{y_1b(y_1^2)}
		\left(\frac{q y_1b(y_1^2)}{g(y_1,y_2)}C^{-,0}-{\cal W}_{11}^{-,0}(y_1,y_2)\right)\ ,
		\eeqa
		where $C^{-,0}$ is a free scalar and ${\cal W}_{11}^{-,0}(y_1,y_2)$ remains undetermined. Using (\ref{w21}),(\ref{w12}), one finds that the nonzero entries of $\widetilde{\cal M}^{-,0}$ are
		given by
		\beqa
		\widetilde{\cal M}^{-,0}_{11}=q y_1 b(y_1^2)Y_-(y_1|y_1)C^{-,0}\,,
		\quad
		\widetilde{\cal M}^{-,0}_{12}=-q y_2 b(y_2^2)Y_-(y_2|y_2)C^{-,0}\,.\label{M12m0}
		\eeqa
		Comparing  (\ref{X0m}) to \eqref{sollin} where 
		$\det\widetilde{\cal M}^{-,0}_{[2,1]}=\widetilde{\cal M}^{-,0}_{12}$
		 we finally identify $\psi^-_0$:
		\beqa
		\psi_0^-=  \frac{ q^2 r_0^2 \tb^2}{b(q)\left(q^2 r_0^2 \tb^2;q^2\right){}_{1}C^{-,0}}\ .\label{psi0m}
		\eeqa
\vspace{1mm}

Now, consider $M=1$. Firstly, one shows
		the identity
		\beqa
		g(y_1,y_2)Y_-(y_1|y_2)+\frac{(q+q^ {-1})^2}{\rho b(y_1^2)b(q^2y_1^2)}\left(\eta^*+\eta\frac{qy_1^2+q^{-1}y_1^{-2}}{q+q^{-1}}\right)-\theta_1=g(y_1,y_2)Y_-(y_2|y_2)\,,\non
		\eeqa
		from which we get
		\beqa
		\widetilde{\cal M}^{-,1}_{11}&=&g(y_1,y_2)\left(Y_-(y_2|y_2){\cal W}_{11}^{-,1}(y_1,y_2)+\frac{y_1b(y_1^2)}{y_2b(y_2^2)}Y_-(y_1|y_1){\cal W}_{12}^{-,1}(y_1,y_2)\right)\,,\non\\
		\widetilde{\cal M}^{-,1}_{12}&=&g(y_2,y_1)\frac{y_2b(y_2^2)}{y_1b(y_1^2)}\left(Y_-(y_2|y_2){\cal W}_{11}^{-,1}(y_1,y_2)+\frac{y_1b(y_1^2)}{y_2b(y_2^2)}Y_-(y_1|y_1){\cal W}_{12}^{-,1}(y_1,y_2)\right)\,,\non\\
		\widetilde{\cal M}^{-,1}_{21}&=&g(y_1,y_2)\left(Y_-(y_2|y_2){\cal W}_{21}^{-,1}(y_1,y_2)+\frac{y_1b(y_1^2)}{y_2b(y_2^2)}Y_-(y_1|y_1){\cal W}_{22}^{-,1}(y_1,y_2)\right)\,,\non\\
		\widetilde{\cal M}^{-,1}_{22}&=&g(y_2,y_1)\frac{y_2b(y_2^2)}{y_1b(y_1^2)}\left(Y_-(y_2|y_2){\cal W}_{21}^{-,1}(y_1,y_2)+\frac{y_1b(y_1^2)}{y_2b(y_2^2)}Y_-(y_1|y_1){\cal W}_{22}^{-,1}(y_1,y_2)\right)\,.\non
		\eeqa
		Conditions (i) imply
		\beqa\label{w21m1}
		\cal W^{-,1}_{21}(y_1,y_2) =-\frac{y_1b(y_1^2)Y_-(y_1|y_1)}{y_2b(y_2^2)Y_-(y_2|y_2)}\cal W^{-,1}_{22}(y_1,y_2)\ 
		\eeqa
		with $\cal W^{-,1}_{22}(y_1,y_2)$ arbitrary.		
		To satisfy conditions (ii), one imposes $\frac{\partial \widetilde{\cal M}^{-,1}_{12}}{ \partial y_1}=\frac{\partial \widetilde{\cal M}^{-,1}_{11}}{\partial y_2}=0$. The associated differential equations,
		omitted here for brevity, are solved by
		\beqa\label{w12m1}
		{\cal W}_{12}^{-,1}(y_1,y_2)= \frac{y_2b(y_2^2)}{Y_-(y_1|y_1)}
		\left(\frac{b(q)}{q^2r_0^3\tb\tb^*\tb^\diamond g(y_1,y_2)}C^{-,1}
		-\frac{Y_-(y_2|y_2)}{y_1b(y_1^2)}{\cal W}_{11}^{-,1}(y_1,y_2)\right)\ 
		\eeqa
		where $C_1$ is a free scalar and ${\cal W}_{11}^{-,1}(y_1,y_2)$ remains undetermined. Using (\ref{w21m1}),(\ref{w12m1}), one finds that the nonzero entries of $\widetilde{\cal M}^{-,1}$ are
		given by
		\beqa
		\widetilde{\cal M}^{-,1}_{11}=\frac{b(q)}{q^2r_0^3\tb\tb^*\tb^\diamond} y_1 b(y_1^2)C^{-,1}\,,
		\quad
		\widetilde{\cal M}^{-,1}_{12}=-\frac{b(q)}{q^2r_0^3\tb\tb^*\tb^\diamond} y_2 b(y_2^2)C^{-,1}\,.\label{M12m1}
		\eeqa
		Comparing $\det\widetilde{\cal M}^{-,1}_{[2,1]}=\widetilde{\cal M}^{-,1}_{12}$
		to (\ref{X1m}) we finally identify $\psi^-_1$:
		\beqa
		\psi_1^-=  -\frac{ q^3r_0^2 \tb^2\tb^*\left(-\frac{q r_0 \tb \tb^{\diamond }}{\tb^*};q^2\right){}_{1}
			\left(-\frac{q r_0 \tb^{\diamond } \tb^*}{\tb};q^2\right){}_{1}}
		{ \left(q^2 r_0^2 \tb^2;q^2\right){}_{1}C^{-,1}}\ .\label{psi1m}
		\eeqa
		\vspace{1mm}

We can now illustrate Corollary \ref{qRacdet} using the results of the two previous examples. From
(\ref{Racdet}), it is clear that $R^{\{.,\diamond\}}_0(\theta^*_0)=R^{\{.,\diamond\}}_0(\theta^*_1)=1$. For $M=1$, using (\ref{M12m0},\ref{psi0m},\ref{M12m1},\ref{psi1m}) in
(\ref{Racdet}), one obtains
\beqa\label{qRacdetexa1}
R^{\{.,\diamond\}}_1(\theta^*_N)&=&\frac{\psi_1^-}{\psi_0^-}
\frac{\det \widetilde{\cal M}_{[2,2]}^{-,1}}{\det \widetilde{\cal M}_{[2,2]}^{-,0}}=
\frac{\psi_1^-}{\psi_0^-}
\frac{\widetilde{\cal M}_{11}^{-,1}}{\widetilde{\cal M}_{11}^{-,0}}\nonumber\\
&=&
-\frac{b(q)^2\left(-\frac{q r_0 \tb \tb^{\diamond }}{\tb^*};q^2\right){}_{1}
			\left(-\frac{q r_0 \tb^{\diamond } \tb^*}{\tb};q^2\right){}_{1}}
			{q^2r_0^3\tb\tb^\diamond}\frac{1}{Y_-(y_1|y_1)}\,,
\eeqa
where $y_1$ is one of the two admissible solutions of the inhomogeneous Bethe equations (\ref{Bfuncinhom}) for $s=1/2$ and the parametrization \eqref{bcdiamzeta}. Recall Example \ref{example3} with $u_1=y_1$. Substituting the expressions (\ref{U1i0},\ref{U1i1}) with $u_1=y_1$ in (\ref{qRacdetexa1}) one obtains
\beqa
R^{\{.,\diamond\}}_1(\theta^*_0)=1\,,\quad
R^{\{.,\diamond\}}_1(\theta^*_1)=
\frac{qr_0(\tb+qr_0\tb^*\tb^\diamond)(\tb^*+qr_0\tb \tb^\diamond)}
{(\tb^\diamond+qr_0\tb\tb^*)(1+q^3r_0^3\tb \tb^*\tb^\diamond)}\ ,
\eeqa
which matches the expected expression for the $q$-Racah polynomial, see (\ref{qracahexampleapp}).

	\vspace{3mm}
$\bullet$
	\underline{The case $s=1/2$, $\epsilon=+$:} Consider $M=0$. Similarly to the previous case,  one firstly
	shows the identity
	\beqa
	\qquad \,\,g(y_1,y_2)Y_+(y_1|y_2)+\frac{(q+q^ {-1})^2}{\rho b(y_1^2)b(q^2y_1^2)}\left(\eta^*+\eta\frac{qy_1^2+q^{-1}y_1^{-2}}{q+q^{-1}}\right)-\theta_0=-\frac{g(y_1,y_2)}{q^5r_0^4\tb\tb^*(\tb^\diamond)^2}Z_1(y_1)Z_2(y_2)\,,
	\eeqa
	where $Z_{1}(y)$ and $Z_{2}(y)$ are given in (\ref{Z1},\ref{Z2}), and then proceed as before. From conditions (i) we find
	\beqa
	{\cal W}_{21}^{+,0}(y_1,y_2)=-\frac{y_1^{-1}b(y_1^2)Z_2(y_1)}{y_2^{-1}b(y_2^2)Z_2(y_2)}
	{\cal W}_{22}^{+,0}(y_1,y_2)
	\eeqa
	and from conditions (ii) it follows
	\beqa
	{\cal W}_{12}^{+,0}(y_1,y_2)=\frac{y_2^{-1}b(y_2^2)}{Z_2(y_1)}
	\left(\frac{q}{g(y_1,y_2)}C^{+,0}-\frac{Z_2(y_2)}{y_1^{-1}b(y_1^2)}{\cal W}_{11}^{+,0}(y_1,y_2)\right)\ ,
\,
	\eeqa
	where $C^{+,0}$ is a free scalar and ${\cal W}_{22}^{+,0}(y_1,y_2)$, ${\cal W}_{11}^{+,0}(y_1,y_2)$
	are arbitrary. Then, one finds the following nonzero  entries of $\widetilde{\cal M}^{+,0}$:
		\beqa
		\widetilde{\cal M}^{+,0}_{11}=-\frac{1}{q^4r_0^4\tb\tb^*(\tb^\diamond)^2} y_1^{-1} b(y_1^2)Z_1(y_1)C^{+,0}\,,
		\quad
		\widetilde{\cal M}^{+,0}_{12}=\frac{1}{q^4r_0^4\tb\tb^*(\tb^\diamond)^2} y_2^{-1} b(y_2^2)Z_1(y_2)C^{+,0}\,.\label{M12p0}
		\eeqa
		Comparing $\det\widetilde{\cal M}^{+,0}_{[2,1]}=\widetilde{\cal M}^{+,0}_{12}$
		to (\ref{X0p}) we finally identify $\psi^+_0$:
		\beqa
		\psi_0^+= -\frac{r_0\tb\tb^\diamond}{C^{+,0}}\ .\label{psi0p}
		\eeqa
		
Now, consider $M=1$. Firstly, one shows
	\beqa
	\,\,g(y_1,y_2)Y_+(y_1|y_2)+\frac{(q+q^ {-1})^2}{\rho b(y_1^2)b(q^2y_1^2)}\left(\eta^*+\eta\frac{qy_1^2+q^{-1}y_1^{-2}}{q+q^{-1}}\right)-\theta_1=-\frac{g(y_1,y_2)}{q^5r_0^4\tb\tb^*(\tb^\diamond)^2}Z_1(y_2)Z_2(y_1)\,
	\eeqa
and proceed as before. From conditions (i) we find
	\beqa
	{\cal W}_{21}^{+,1}(y_1,y_2)=-\frac{y_1^{-1}b(y_1^2)Z_1(y_1)}{y_2^{-1}b(y_2^2)Z_1(y_2)}
	{\cal W}_{22}^{+,1}(y_1,y_2)
	\eeqa
	and from conditions (ii) it follows
	\beqa
	{\cal W}_{12}^{+,1}(y_1,y_2)=\frac{y_2^{-1}b(y_2^2)}{Z_1(y_1)}
	\left(\frac{q}{g(y_1,y_2)}C^{+,1}-\frac{Z_1(y_2)}{y_1^{-1}b(y_1^2)}{\cal W}_{11}^{+,1}(y_1,y_2)\right)
\,
	\eeqa
	where $C^{+,1}$ is a free scalar and ${\cal W}_{22}^{+,1}(y_1,y_2)$, ${\cal W}_{11}^{+,1}(y_1,y_2)$
	are arbitrary. Then, one finds the following nonzero entries of $\widetilde{\cal M}^{+,1}$:
		\beqa
		\widetilde{\cal M}^{+,1}_{11}=-\frac{1}{q^4r_0^4\tb\tb^*(\tb^\diamond)^2} y_1^{-1} b(y_1^2)Z_2(y_1)C^{+,1}\,,
		\quad
		\widetilde{\cal M}^{+,1}_{12}=\frac{1}{q^4r_0^4\tb\tb^*(\tb^\diamond)^2} y_2^{-1} b(y_2^2)Z_2(y_2)C^{+,1}\,.\label{M12p1}
		\eeqa
		Comparing $\det\widetilde{\cal M}^{+,1}_{[2,1]}=\widetilde{\cal M}^{+,1}_{12}$
		to (\ref{X1p}) we finally identify $\psi^+_1$:
		\beqa
		\psi_1^+= -\frac{r_0\tb^\diamond}{C^{+,1}}\ .\label{psi1p}
		\eeqa

\section{Concluding remarks}
 In this paper, the Leonard pair/algebraic Bethe ansatz correspondence  initiated in \cite{BP19,BP22} is further studied. The main results are explicit expressions for normalized scalar products of on-shell versus off-shell Bethe states in terms of $q$-Racah polynomials, see Theorem \ref{prop:scalprod}. This extends the results of \cite{BP22} obtained for the on-shell versus on-shell Bethe states. Upon  specializations,  the normalized scalar products here computed are shown to satisfy a Belliard-Slavnov system of linear equations, see Theorem \ref{thm:BSscal}. This implies  the existence of a determinant formula for $q$-Racah polynomials, which involves Bethe roots satisfying a Bethe equation of inhomogeneous type; See Corollary \ref{qRacdet}. Also, a set of relations that determine the solutions of Bethe equations of inhomogeneous type in terms of solutions of Bethe equations of homogeneous type is obtained, see Proposition \ref{relinhhom}. \vspace{1mm}

The results obtained show that  the finite dimensional representation theory of the Askey-Wilson algebra (i.e. the theory of Leonard pairs and triples) gives: (i) a way to compute normalized scalar products of on-shell and off-shell Bethe states that solve a class of Belliard-Slavnov systems of linear equations; (ii) a way to relate a class of inhomogenous Bethe equations to homogenous Bethe equations by specialization of scalar products. Reciprocally, the approach followed leads to (iii) a new interpretation of   discrete ($q$-Racah) orthogonal polynomials of the Askey-scheme as determinants of matrices associated with Bethe roots of inhomogeneous type.\vspace{1mm}

Some perspectives are now presented.
Firstly, it is clear that the Leonard pair/algebraic Bethe ansatz correspondence for the Racah case ($q=1$)  can be completed along the same lines  starting from the material  in \cite{Nico1,Nico2}.
Secondly, it is expected that the Leonard pair/algebraic Bethe ansatz correspondence generalizes to a tridiagonal pair/algebraic Bethe ansatz correspondence. Indeed, 
it  is known that the Askey-Wilson algebra with generators $\tA,\tA^*$ is a quotient of the $q$-Onsager algebra $O_q$ with generators $\tW_0,\tW_1$. From the point of view of the algebraic Bethe ansatz, it is possible to show that the exchange relations (\ref{comBdBd})-(\ref{comDdCd}) are solved by dynamical operators expressed in terms of  $\tW_0,\tW_1$. Now, recall that irreducible finite dimensional representations of $O_q$ are classified according to the theory of tridiagonal pairs of $q$-Racah type \cite{IT09} that generalizes the theory of Leonard pairs.  Similarly to the Leonard pair/algebraic Bethe ansaz correspondence, a tridiagonal pair/algebraic Bethe ansatz correspondence is thus expected. From the point of view of representation theory and special functions, this correspondence should lead to some interesting results. For instance, for a tridiagonal pair of $q$-Racah type $(\pi(\tW_0),\pi(\tW_1))$ with $\pi: O_q \rightarrow {\rm End}(\cV)$, a characterization of the entries of the transition matrices relating certain eigenbases of $\pi(\tW_0)$ to certain eigenbases of $\pi(\tW_1)$ in terms of special functions is an open problem - see however \cite{BVZ16} for some examples, and \cite{CGT} for the case of tridiagonal pairs of type II. Within a tridiagonal pair/algebraic Bethe ansatz correspondence,  we expect the special functions of interest - generalizing the $q$-Racah polynomials - can be computed in terms of normalized scalar products of on-shell Bethe states.
From the point of view of quantum integrable models, a tridiagonal pair/algebraic Bethe ansatz correspondence should also give new insights for the open XXZ spin chain with integrable boundary conditions or certain partition functions of the six-vertex model generalizing \cite{BPS24}. In particular, for the case of special right (or left) boundary conditions for the open XXZ spin chain, the spectral problem for the Hamiltonian should admit a solution either associated with Bethe equations of homogeneous type, or associated with Bethe equations of inhomogenous type. In both cases, the on-shell Bethe states diagonalizing the Hamiltonian should  find an interpretation in terms of eigenvectors of   $\pi(\tW_0)$ (or $\pi(\tW_1)$).  Related scalar products should also find an interpretation in terms of special functions generalizing the $q$-Racah polynomials. 

Some of these problems will be discussed elsewhere. 

\vspace{0.2cm}

\noindent{\bf Acknowledgments:} We thank S. Belliard for discussions and N. Crampé for comments on the manuscript.
P.B.  is supported by C.N.R.S. R.A.P. thanks the Institut Denis-Poisson for hospitality.

\begin{appendix}

\section{Coefficients of the exchange relations}\label{apA}
Let $\alpha,\beta \in {\mathbb C}^*$ or  $\alpha=0,\beta \in {\mathbb C}^*$ or  $\alpha \in {\mathbb C}^*,\beta=0$. The coefficients of the exchange relations (\ref{comBdBd})-(\ref{comDdCd})  are given by
\ben
\nonumber&&f(u,v)= \frac{b(qv/u)b(uv)}{b(v/u)b(quv)}\,,\quad h(u,v)= \frac{b(q^2uv)b(qu/v)}{b(quv)b(u/v)},\\
\nonumber&&g(u,v,m)=\frac{\gamma(u/v,m+1)}{\gamma(1,m+1)}\frac{b(q) b\left(v^2\right)}{b\left(q v^2\right) b\left(\frac{u}{v}\right)},
\quad w(u,v,m)=-\frac{\gamma(uv,m)}{\gamma(1,m+1)}\frac{b(q)}{b(q u v)},\\
\nonumber &&k(u,v,m)=\frac{ \gamma(v/u,m+1)}{\gamma(1,m+1)}\frac{b(q) b\left(q^2 u^2\right)}{b\left(q u^2\right) b\left(\frac{v}{u}\right)}, \quad
n(u,v,m)=\frac{\gamma(1/(uv),m+2)}{\gamma(1,m+1)} \frac{b(q) b\left(v^2\right) b\left(q^2 u^2\right)}{b\left(q u^2\right) b\left(q v^2\right)
	b(q u v)}\,\,,
\een
\ben
&&\qquad q(u,v,m)=\frac{ \gamma \left(u/v,m\right)b(q) b(u v)}{\gamma (1,m+1) b\left(u/v\right) b(q u v)}\,,\quad
r(u,v,m)=\frac{b(q) b\left(u^2\right) \gamma (1,m) \gamma \left(v/u,m+1\right)}{\gamma (1,m+1)^2 b\left(q
	u^2\right) b\left(v/u\right)}\,,\nonumber\\
\nonumber&&s(u,v,m)=\frac{b(q)^2 b\left(u^2\right) \gamma \left(v^{-2},m+1\right) \gamma \left(v/u,m+1\right)}{\gamma
	(1,m+1)^2 b\left(q u^2\right) b\left(q v^2\right) b\left(\frac{v}{u}\right)}\,,
\quad
x(u,v,m)=\frac{b(q) b\left(u^2\right) b\left(q u/v\right) \gamma \left(1/(uv),m+1\right)}{\gamma (1,m+1)
	b\left(q u^2\right) b\left(u/v\right) b(q u v)}\,,\\
\nonumber&&y(u,v,m)=-\frac{b(q)^2 \gamma \left(v^{-2},m+1\right) \gamma (u v,m)}{\gamma (1,m+1)^2 b\left(q v^2\right) b(q u v)}\,,
\quad
z(u,v,m)=-\frac{b(q) \gamma (1,m) \gamma (u v,m)}{\gamma (1,m+1)^2 b(q u v)}\,.\nonumber
\een
where
\ben
\gamma^\epsilon (u,m)= \alpha ^{\frac{1-\epsilon }{2}} \beta ^{\frac{\epsilon +1}{2}}
q^{-m} u -\alpha ^{\frac{\epsilon +1}{2}} \beta ^{\frac{1-\epsilon }{2}}
q^m u^{-1} \ .\label{gam}
\een

\section{Example of a Leonard triple and eigenbases for $s=1/2$}\label{apB}
\subsection{A Leonard triple from a spin-$s$ representation of $U_q(sl_2)$} 
Recall the quantum algebra $U_q(sl_2)$ with generators $S_\pm,s_3$ and defining relations:
\ben\label{uqsl2}
\left[s_3,S_{\pm}\right]=\pm S_{\pm}\,,\quad \left[S_+,S_-\right]=\frac{q^{2s_3}-q^{-2s_3}}{q-q^{-1}}\,.
\een
Now let $(\pi,{\cal V})$ denote the irreducible (spin-$s$) representation of $U_q(sl_2)$  of dimension $\dim({\cal V})=2s+1$  on which the generators $\{q^{\pm s_3},S_\pm\}$ act as:
\ben
\pi (q^{\pm s_3}) = \sum_{k=1}^{2s+1}q^{\pm ( s+1-k)}E_{kk}\,,\   \pi (S_-) = \sum_{k=1}^{2s}  \sqrt{[k]_q[2s+1-k]_q} E_{k+1k}\, , \  \pi (S_+) = \sum_{k=1}^{2s}  \sqrt{[k]_q[2s+1-k]_q} E_{kk+1}\, ,\nonumber
\een
where the matrix $E_{ij}$ has a unit at the intersection of the $i-$th row and the $j-$th column, and all other entries are zero. The following gives an example of Leonard triple $A,A^*,A^\diamond$ of $q$-Racah type satisfying the Askey-Wilson relations (\ref{AWtriple}):
\beqa
\quad A &=&  - (q-q^{-1})r_0^{-1/2}(\tb^\diamond)^{1/2} q^{s+1/2} \pi(S_+q^{s_3}) + (q-q^{-1})r_0^{-1/2}(\tc^\diamond)^{1/2} q^{-s-1/2} \pi(S_-q^{s_3}) + \theta_{s} \pi(q^{2s_3})\ ,\label{AtripleUqsl2}\nonumber\\
 \quad A^*&=& - (q-q^{-1})r_0^{-1/2}(\tc^\diamond)^{1/2} q^{-s-1/2} \pi(S_+q^{-s_3}) + (q-q^{-1})r_0^{-1/2}(\tb^\diamond)^{1/2} q^{s+1/2} \pi(S_-q^{-s_3}) + \theta^*_{s} \pi(q^{-2s_3})\label{AstartripleUqsl2}\ .\nonumber
\eeqa
The expression for $A^\diamond$ is computed using (\ref{Adiamond}). One gets\footnote{For simplicity,
we choose the branch cut $\sqrt{1/(r_0^2\tb^\diamond)}=1/(r_0\sqrt{\tb^\diamond})$.}:
\beqa
&&A^\diamond =  \left(q^{2s-1}\tb^\diamond-q^{1-2s}\tc^\diamond\right)\pi\left(\frac{q^{2s_3}-q^{-2s_3}}{q+q^{-1}}\right)+
r_0^{-1}(q-q^{-1})^2 \pi(S_-^2)
-
\frac{(q-q^{-1})^2}{q+q^{-1}}\theta_s^\diamond \pi( S_-S_+)
\nonumber\\&&\qquad
+(q-q^{-1})r_0^{1/2}
\pi(S_-)\left(q^{-s+1/2}(\tc^\diamond)^{1/2}\theta_s^*\pi(q^{-s_3})
+q^{s-1/2}(\tb^\diamond)^{1/2}\theta_s \pi(q^{s_3})\right)\nonumber\\&&\qquad
+
\frac{r_0}{q+ q^{-1}}\left(\theta_s\theta_s^*+\frac{\omega^{\{.,*,\diamond\}} }{(q-q^{-1})^2}
\right){\mathbb I} \ .\label{AdiamtripleUqsl2}\nonumber
\eeqa
\begin{rem}
With respect to the conventions chosen in \cite{BP19}, we have the following identifications:
\beqa
&&\tb \rightarrow \frac{1}{2}e^{-\mu}q^{\frac{1}{2}(\nu+\nu')-2s}\ , \quad
\tb^* \rightarrow \frac{1}{2}e^{\mu'}q^{\frac{1}{2}(\nu+\nu')-2s}\ ,
\quad \tb^\diamond  \rightarrow \frac{1}{2}v^2q^{\frac{1}{2}(\nu+\nu')-2s} \ , \nonumber\\
&&\tc \rightarrow \frac{1}{2}e^{\mu}q^{\frac{1}{2}(\nu+\nu')+2s}\ ,
\quad \tc^* \rightarrow \frac{1}{2}e^{-\mu'}q^{\frac{1}{2}(\nu+\nu')+2s}\ , \quad
\tc^\diamond  \rightarrow \frac{1}{2}v^{-2}q^{\frac{1}{2}(\nu+\nu')+2s}\ , \nonumber\\
&& r_0 \rightarrow 2q^{-\frac{1}{2}(\nu+\nu')}\ ,
\quad v \rightarrow \zeta^{-1}\ ,\quad \chi \rightarrow -\frac{1}{2}q^\nu(q^2-q^{-2})\ ,
\eeqa
where $\{\mu,\mu',\nu,\nu',v\}$ are complex parameters.
\end{rem}

\subsection{Eigenvectors and their duals for $s=1/2$.} In this
subsection, as a simple example, the eigenvectors of the Leonard triple $(A,A^*,A^\diamond)$ 
and respective normalizations for the case $s=1/2$ are given.

From (\ref{AtripleUqsl2}), the eigenvectors of $A$ are
\beqa
|\theta_0\rangle=\mathfrak{n}_0\left(
\begin{array}{c}
	\frac{q^{-1/2}r_0^{-1/2}(\tb^\diamond)^{1/2}}{\tb}\\
	1 \\
\end{array}
\right)\,,\quad
|\theta_1\rangle=\mathfrak{n}_1
\left(
\begin{array}{c}
	q^{3/2}r_0^{3/2}(\tb^\diamond)^{1/2}\tb \\
	1 \\
\end{array}
\right) \ ,
\eeqa 
where $\mathfrak{n}_{0,1}$ are some scalars.
 The action
(\ref{tridAstar}) fixes the ratio
\beqa
\frac{\mathfrak{n}_1}{\mathfrak{n}_0}=-\frac{\tb^\diamond+qr_0\tb\tb^*}
{qr_0\tb(\tb^*+qr_0\tb\tb^\diamond)}\ .
\eeqa
For the dual eigenvectors, one has
\beqa
\langle \theta_0|=\mathfrak{m}_0\left(
\begin{array}{cc}
	1\,, & -q^{3/2}r_0^{3/2}(\tb^\diamond)^{1/2}\tb\\
\end{array}
\right)\ ,
\quad
\langle \theta_1|=\mathfrak{m}_1
\left(
\begin{array}{cc}
	1\,, & -\frac{q^{-1/2}r_0^{-1/2}(\tb^\diamond)^{1/2}}{\tb}\\
\end{array}
\right)\ ,
\eeqa
where $\mathfrak{m}_{0,1}$ are some scalars. The dual action (\ref{tridAstardual}) fixes
\beqa
\frac{\mathfrak{m}_1}{\mathfrak{m}_0}=
\left(\frac{\xi_1}{\xi_0}\right)
\frac{qr_0\tb(\tb^*+qr_0\tb \tb^\diamond)}{(\tb^\diamond+qr_0\tb \tb^*)}\, .
\eeqa
Choosing
\beqa
\mathfrak{m}_0=\frac{q^{1/2}r_0^{1/2}\tb\xi_0}
{(\tb^\diamond)^{1/2}(1-q^2r_0^2\tb^2)\mathfrak{n}_0}\ ,
\eeqa
one has - see \eqref{orthcond}:
\beqa
\langle \theta_0|\theta_0\rangle=\xi_0\,,\quad \langle \theta_1|\theta_1\rangle=\xi_1\ .
\eeqa

From (\ref{AstartripleUqsl2}), the eigenvectors of $A^*$ are
\beqa
|\theta_0^*\rangle=\mathfrak{n}_0^*\left(
\begin{array}{c}
	-\frac{q^{1/2}r_0^{1/2}\tb^*}{(\tb^\diamond)^{1/2}}\\
	1 \\
\end{array}
\right)\,,\quad
|\theta_1^*\rangle=\mathfrak{n}_1^*
\left(
\begin{array}{c}
	-\frac{1}{q^{3/2}r_0^{3/2}(\tb^\diamond)^{1/2}\tb^*} \\
	1 \\
\end{array}
\right) \ ,
\eeqa 
where $\mathfrak{n}^*_{0,1}$ are some scalars. The action
(\ref{tridAstar2}) fixes
\beqa
\frac{\mathfrak{n}_1^*}{\mathfrak{n}_0^*}=
-\frac{qr_0\tb^*(\tb^\diamond+qr_0\tb \tb^*)}{(\tb+qr_0\tb^*\tb^\diamond)}
 \ .
\eeqa
For the dual vectors, one has
\beqa
\langle \theta_0^*|=\mathfrak{m}_0^*\left(
\begin{array}{cc}
	1\,, & \frac{1}{q^{3/2}r_0^{3/2}(\tb^\diamond)^{1/2}\tb^*}\\
\end{array}
\right)\ ,\quad
\langle \theta_1^*|=\mathfrak{m}_1^*
\left(
\begin{array}{cc}
	1\,, & \frac{q^{1/2}r_0^{1/2}\tb^*}{(\tb^\diamond)^{1/2}}\\
\end{array}
\right)\ ,
\eeqa
where $\mathfrak{m}_{0,1}^*$ are some scalars. The dual action (\ref{tridAstar2dual}) fixes
\beqa
\frac{\mathfrak{m}_1^*}{\mathfrak{m}_0^*}=
\left(\frac{\xi_1^*}{\xi_0^*}\right)
\frac{(\tb+qr_0\tb^* \tb^\diamond)}{qr_0\tb^*(\tb^\diamond+qr_0\tb \tb^*)}
\ .
\eeqa
Choosing
\beqa
\mathfrak{m}_0^*=\frac{q^{3/2}r_0^{3/2}(\tb^\diamond)^{1/2}\tb^*\xi_0^*}
{(1-q^2r_0^2(\tb^*)^2)\mathfrak{n}_0^*}\ ,
\eeqa
one has
\beqa
\langle \theta_0^*|\theta_0^*\rangle=\xi_0^*\,,\quad \langle \theta_1^*|\theta_1^*\rangle=\xi_1^*\ .
\eeqa

From (\ref{AdiamtripleUqsl2}), the eigenvectors of $A^\diamond$ are
\beqa
|\theta_0^\diamond\rangle=\mathfrak{n}_0^\diamond\left(
\begin{array}{c}
	0\\
	1 \\
\end{array}
\right)\,,\quad
|\theta_1^\diamond\rangle=\mathfrak{n}_1^\diamond
\left(
\begin{array}{c}
 \frac{q^{1/2} r_0^{1/2} \tb \tb^* \left(q^2 r_0^2 (\tb^{\diamond })^2-1\right)}
 {(\tb^{\diamond})^{1/2} \left(q^2 r_0^2 \tb \tb^* \left(\tb^*+q r_0 \tb \tb^{\diamond }\right)+\tb+q r_0 \tb^{\diamond } \tb^*\right)} \\
 1 \\
\end{array}
\right) \ ,
\eeqa 
where $\mathfrak{n}_{0,1}^\diamond$ are some scalars. The action
(\ref{tridAc}) for $(b,c)=(*,\diamond)$ fixes the ratio
\beqa
\frac{\mathfrak{n}_1^\diamond}{\mathfrak{n}_0^\diamond}= 
-\frac{q^2 r_0^2 \tb \tb^* \left(\tb^*+q r_0 \tb \tb^{\diamond }\right)+
\tb+q r_0 \tb^{\diamond } \tb^*}{q r_0 \left(\tb^*+q r_0 \tb \tb^{\diamond }\right)
\left(\tb^{\diamond }+q r_0 \tb \tb^*\right)}\ .
\eeqa

The transitions (\ref{transeqtrip}) and (\ref{ctob}) for $(a,b,c)=(.,*,\diamond)$
determine the ratios
\beqa
\frac{\mathfrak{n}_0^*}{\mathfrak{n}_0^\diamond}=\frac{1}{g_0(1-q^2r_0^2(\tb^*)^2)}\ ,
\quad \frac{\mathfrak{n}_0}{\mathfrak{n}_0^\diamond}=
-\frac{q r_0 \tb \left(\tb^*+q r_0 \tb \tb^{\diamond }\right)}{g_0 \tb^{\diamond } \left(q^2 r_0^2 \tb^2-1\right) (q^2 r_0^2 \left(\tb^*\right)^2-1)}
\eeqa
while (\ref{ctoa}) gives (\ref{fgh0}) for $s=1/2$:
\beqa
f_0h_0^{-1}g_0^{-1}=
-\frac{(1-q^2 r_0^2 \tb^2) (1-q^2 r_0^2 (\tb^{\diamond })^2) (1-q^2 r_0^2 (\tb^*)^2)}
{\left(\frac{q r_0 \tb \tb^{\diamond }}{\tb^*}+1\right) \left(\frac{q r_0 \tb \tb^*}{\tb^{\diamond }}+1\right) \left(\frac{q r_0 \tb^{\diamond } \tb^*}{\tb}+1\right)}\ .
\eeqa

\subsection{Examples of scalar products}
For the normalized scalar products 
\beqa
\frac{\langle \theta_M |\theta_N^* \rangle}{\langle \theta_M |\theta_M \rangle}\ ,\label{exa}
\eeqa
we obtain:
\beqa
\frac{\langle \theta_0|\theta_0^* \rangle}{\langle \theta_0 |\theta_0\rangle}=1,
\quad
\frac{\langle \theta_0|\theta_1^* \rangle}{\langle \theta_0 |\theta_0\rangle}=-
\frac{(\tb^\diamond+qr_0\tb \tb^*)(1+q^3r_0^3\tb \tb^*\tb^\diamond)}
{qr_0(\tb+qr_0\tb^*\tb^\diamond)(\tb^*+qr_0\tb \tb^\diamond)},\quad
\frac{\langle \theta_1|\theta_0^* \rangle}{\langle \theta_1 |\theta_1\rangle}=1, \quad
\frac{\langle \theta_1|\theta_1^* \rangle}{\langle \theta_1 |\theta_1\rangle}=-1\ .\label{ratex}
\eeqa

From those expressions, the $q$-Racah polynomials can be computed (see \cite[Theorem 14.6 and 15.6]{T03}):
\beqa
R^{\{.,\diamond\}}_M(\theta^*_N) &=& \frac{\langle \theta_M |\theta_N^* \rangle}{\langle \theta_0 |\theta^*_N \rangle}\frac{\langle \theta_0 |\theta_0 \rangle}{\langle \theta_M |\theta_M \rangle}=
\frac{\langle \theta_N^* |\theta_M \rangle}{\langle \theta_0^* |\theta_M \rangle}
\frac{\langle \theta_0^* |\theta_0^* \rangle}{\langle \theta_N^* |\theta_N^* \rangle}
\ .\label{Rac1bis}
\eeqa
 Using \eqref{ratex}, it follows $ R^{\{.,\diamond\}}_0(\theta^*_0)= R^{\{.,\diamond\}}_0(\theta^*_1)= R^{\{.,\diamond\}}_1(\theta^*_0)=1$ and
 \beqa\label{qracahexampleapp}
R^{\{.,\diamond\}}_1(\theta^*_1)&=&\frac{\langle \theta_1 |\theta_1^* \rangle}{\langle \theta_0 |\theta^*_1 \rangle}\frac{\langle \theta_0 |\theta_0 \rangle}{\langle \theta_1 |\theta_1 \rangle}=
\frac{\langle \theta_1^*|\theta_1\rangle \langle \theta_0^*|\theta_0^*\rangle}
{\langle \theta_0^*|\theta_1\rangle \langle \theta_1^*|\theta_1^*\rangle}=
\frac{qr_0(\tb+qr_0\tb^*\tb^\diamond)(\tb^*+qr_0\tb \tb^\diamond)}
{(\tb^\diamond+qr_0\tb\tb^*)(1+q^3r_0^3\tb \tb^*\tb^\diamond)}\ 
\non\\ 
&=&1-\frac{(q^{-2},\frac{\tb}{\tc}q^2;q^2)_1q^2}{(\frac{-\tb\tc^\diamond}{\tc^*}r_0q,-\frac{\tb^*\tb^\diamond}{\tc}r_0q^{3},q^{-2};q^2)_1(q^2;q^2)_1\tc^*}\left(\theta_1^* - \tb^* - \tc^*\right)\ ,
\eeqa
where the second line matches 
with (\ref{qracahpoly}) for $M=N=1$ and $s=1/2$.

\end{appendix}

\end{document}